\documentclass[format=acmsmall, review=false, screen=false, manuscript]{acmart}  	% use "amsart" instead of "article" for AMSLaTeX format

\usepackage{amssymb}
\usepackage{amssymb,amsmath,amsthm,amsfonts}
\usepackage{enumitem}
\usepackage{amssymb}
\usepackage{amsfonts}
\usepackage{color}
\usepackage{dsfont}
\usepackage{graphicx}
\usepackage[normalsize]{subfigure}
\usepackage{epstopdf}
\usepackage{array}
\usepackage{color}
\newcommand{\argmin}{\operatornamewithlimits{argmin}}

\newcolumntype{L}{>{\arraybackslash}m{11.75cm}} %11.75 is full page

\newcommand{\E}[1]{\mathbb{E}\left[#1\right]} % Expected Value

\newcommand{\Var}[1]{\operatorname{Var}(#1)}  %Variance 
\newcommand{\gstar}{$\mbox{GREEDY}^*$}
\newcommand{\gstars}{\mbox{GREEDY}^*}

\newcommand{\gs}{G^*}
\theoremstyle{definition}
\newtheorem{remark}{Remark}
\newtheorem{observation}{Observation}

\allowdisplaybreaks

\begin{document}
\title[Towards Optimality in Parallel Job Scheduling]{Towards Optimality in Parallel Job Scheduling}

\author{ Benjamin Berg}
\affiliation{Carnegie Mellon University}
\email{bsberg@cs.cmu.edu}

\author{Jan-Pieter Dorsman}
\affiliation{University of Amsterdam}
\email{j.l.dorsman@uva.nl}

\author{Mor Harchol-Balter}
\affiliation{Carnegie Mellon University}
\email{harchol@cs.cmu.edu}

\begin{abstract}
To keep pace with Moore's law, chip designers have focused on
increasing the number of cores per chip rather than single core
performance.  In turn, modern jobs are often designed to run on any
number of cores.  However, to effectively leverage these multi-core chips, one must address the question of how many cores to assign to each job.  Given that jobs receive sublinear speedups from
additional cores, there is an obvious tradeoff: allocating more cores
to an individual job reduces the job's runtime, but in turn decreases the efficiency of the overall system.  We ask how the system should schedule jobs across cores so as to minimize the mean response
time over a stream of incoming jobs.

To answer this question, we develop an analytical model of jobs
running on a multi-core machine.  We prove that EQUI, a
policy which continuously divides cores evenly across jobs, is optimal
when all jobs follow a single speedup curve and have exponentially distributed sizes.  EQUI requires jobs to change their level of parallelization while they run.  Since this is not possible for all workloads, we consider a class of ``fixed-width'' policies, which choose a single level of parallelization, $k$, to use for all jobs.  We prove that, surprisingly, it is possible to achieve EQUI's performance without requiring jobs to change their levels of parallelization by using the \emph{optimal} fixed level of parallelization, $k^*$.  We also show how to analytically derive the optimal $k^*$ as a function of the system load, the speedup curve, and the job size distribution.  

In the case where jobs may follow different speedup curves, finding a good scheduling policy is even more challenging.  In particular, we find that policies like EQUI which performed well in the case of a single speedup function now perform poorly.  We propose a very simple policy, $\mbox{GREEDY}^*$, which performs near-optimally when compared to the numerically-derived optimal policy.
\end{abstract}
\maketitle

\renewcommand{\shortauthors}{B. Berg et al.}

\vspace{-.1in}
\section{Introduction}
\subsection*{The Parallelization Tradeoff}

Modern multicore chips afford the opportunity
to reduce job response time by parallelizing a job across several
cores.  To exploit this opportunity, modern parallel jobs are often designed to have the ability to run on any number of cores \cite{srinivasan2003robust, Feitelson:1997:TPP:646378.689517}.  However, effectively exploiting parallelism is non-trivial.
Specifically, for each arriving job one must decide its \emph{level of
parallelization} -- the number of cores on which the job is run.  We consider the setting where jobs arrive over time and ask how to minimize the mean response time across jobs.

In choosing the optimal levels of parallelization, we must consider
the following {\em tradeoff}.  Parallelizing an individual job across
multiple cores reduces the response time of that individual job.  In
practice, however, a job's speedup is sublinear and
concave in its level of parallelization, leading to an inefficient use
of resources and additional system load.  Hence, while a higher level
of parallelization may decrease an individual job's response time, it
may have a deleterious effect on overall mean response time across jobs.  Thus, while traditionally the programmer chooses the level of parallelization for their jobs, we propose making this a system level decision, minimizing mean response time across \emph{all} jobs.  

Access to additional cores can reduce a job's service requirement in a variety of ways.  If a job is composed of many independent computations, as is the case for jobs composed largely of matrix operations, having additional cores can significantly reduce response time.  However, even memory-bound jobs can benefit from more cores since additional cores often provide the job with additional memory throughput.  The effect of parallelization can, however, be limited by several factors:  jobs may require mostly sequential computation, and jobs may incur overhead due to checkpointing and communication between threads.  As the number of cores in the system increases, \cite{Bienia:2008:PBS:1454115.1454128}
 has shown that all the benefits and limitations of parallelization can be encapsulated in a job's \emph{speedup function}, $s(k)$, which specifies the ratio of a job's runtime on $k$ cores to its runtime on a single core.

In general, it is conceivable that every job will have a different speedup function.  However, there are also many workloads \cite{Bienia:2008:PBS:1454115.1454128} %try to find different citation
where all jobs have the same speedup function, for example when one is running multiple instances of the same program.  It turns out that even scheduling jobs with a single speedup function is non-trivial.  Hence, for the majority of the paper we will focus on a single speedup function before turning our attention to multiple speedup functions.

The question of how to choose the correct level of parallelization for jobs is commonly referred to as the choice between \emph{fine grained parallelism}, where every job is parallelized across a large number of cores, and \emph{coarse grained parallelism}, where the level of parallelization is small.  This same tradeoff arises in many systems beyond multicore chips.  For example, in a server farm, jobs could be run across multiple servers.  Running on more servers allows an individual job to complete more quickly, but leads to an inefficient use of resources which could increase the response times of subsequent jobs.  Additionally, an operating system might choose how to partition memory (cache space) between multiple applications.  Likewise, a system designer may have to choose between fewer wider bus lanes for memory access, or several narrower bus lanes.  In all cases, one must balance the effect on an individual job's response time with the effect on overall mean response time.  Throughout this paper, we will use the terminology of parallelizing \emph{jobs} across \emph{cores}, however all of our remarks can be applied equally to any setting where limited resources must be shared amongst concurrently running processes. 
\vspace{-.1in}
\subsection*{Prior Work}
Prior work on exploiting parallelism has traditionally been split between several disparate communities.  The SIGMETRICS/Performance community frequently considers scheduling a stream of incoming jobs for execution across several servers with the goal of minimizing mean response time, e.g., \cite{harchol2009surprising,Nelson1993PerfromEval,lu2011join,gupta2007analysis,tsitsiklis2011power}.  Job sizes are assumed to be random variables drawn from some general distribution, and arrivals are assumed to occur according to some stochastic arrival process.  Typically, it is assumed that each individual job is run on just \emph{one server}: jobs are \emph{not} parallelizable.  An exception is the work on Fork-Join queues.  Here it is assumed that every job is parallelized across \emph{all} servers \cite{ko2004response}; the question of finding the \emph{optimal} level of parallelization is not addressed.  This community has also considered systems where the service rate that a job receives can be adjusted over time, but this work has focused on balancing a tradeoff between response time and some other variable such as power consumption \cite{gandhi2009optimal} or \emph{Goodness of Service} \cite{chaitanya2008qdsl}.

By contrast, the SPAA/STOC/FOCS community extensively studies the effect of parallelism, however it largely focuses on a \emph{single job}.  More recently, this community has considered the effect of parallelization on the mean response time of a stream of jobs.  However, this has been through the lens of competitive analysis, which assumes the job sizes (service times), arrival times, and even speedup curves are adversarially chosen, e.g.\ \cite{Edmonds1999SchedulingIT, edmonds2009scalably,Agrawal:2016:SPJ:2935764.2935782}.  Competitive analysis also does not yield closed form expressions for mean response time.  We seek to analyze this problem using a more typical queueing theoretic model.  Our workloads are drawn from distributions and our analysis yields closed-form expressions for mean response time as well as expressions for the optimal level of parallelization.

The advent of moldable jobs which can run on any number of cores has pushed the high performance computing (HPC) systems community to consider how to effectively allocate cores to jobs \cite{Cirne:2002:UMI:634535.634539, ANASTASIADIS1997109, huang2013effective, srinivasan2003robust}.  These studies tend to be empirical rather than analytical, each looking at specific workloads and architectures, often leading to conflicting results between studies.  By contrast, our paper attacks this problem from a stochastic, analytical point of view, deriving optimal scheduling algorithms in a more general model.
\vspace{-.1in}
\subsection*{Our Model}

We assume that jobs arrive into an $n$-core machine according to a
Poisson Process with rate $\Lambda$ jobs/second where $$\Lambda :=
\lambda n$$ for some $\lambda$.  The random variable $X$ denotes the
\emph{size} of a job.  We think of $X$ as the inherent work associated
with a job.  Any job can be run on any subset of cores.  If a job of
size $X$ is parallelized across $k$ cores, its service requirement becomes $X_k$, where 
\begin{equation}\label{eq:service}
X_k=\frac{X}{s(k)}
\end{equation}
and $s(k) \leq k$ is called the \emph{speedup factor}.  The function
$s:\mathbb{R_+} \rightarrow \mathbb{R}_{+}$ is assumed to be \emph{non-decreasing} and \emph{concave}, in agreement with functions described in \cite{Hill:2008:ALM:1449375.1449387}.  Additionally, we will focus on instances where jobs receive an imperfect speedup and thus $s(k)$ is assumed to be \emph{sublinear}: $s(k) < k$ for all $k>1$ and  there exists some constant $c>0$ such that $s(k) < c$ for all $k\geq0$.  Importantly, we assume that a job is split perfectly so that its service requirement on each of the $k$ cores is the same.  While this assumption ignores the \emph{straggler problem}, in which some pieces of a job may take longer to complete than others, it is worth noting that extensive work has been done on mitigating these effects in practice \cite{ananthanarayanan2014grass,ren2015hopper}.

The optimality proofs in the paper will require assuming that $X$ is exponentially distributed.  However, all of the performance analysis presented admits general job size distributions.

An example of a well-known speedup function is Amdahl's
law~\cite{mccool2012structured}, which models every job as having a fraction of work, $p$, which is parallelizable.  The speedup factor $s(k)$ is then a function of the parameter $p$ as follows:
$$s(k) = \frac{1}{\frac{p}{k} +1 -p}.$$
Figure \ref{figAmdahl} shows Amdahl's law under various values of $p$.  Although Amdahl's law ignores aspects of job behavior, we see in \ref{figParsec} that several workloads from the PARSEC-3 benchmark \cite{zhan2017parsec3} follow speedup curves which can be accurately modeled by Amdahl's law.  Hence, although our analysis will not rely on the specifics of the speedup function, we will use Amdahl's law in numerical examples.
\begin{figure}[ht]
\vspace{-.25in}
\centering
\subfigure[]{
\includegraphics[width=.4\textwidth]{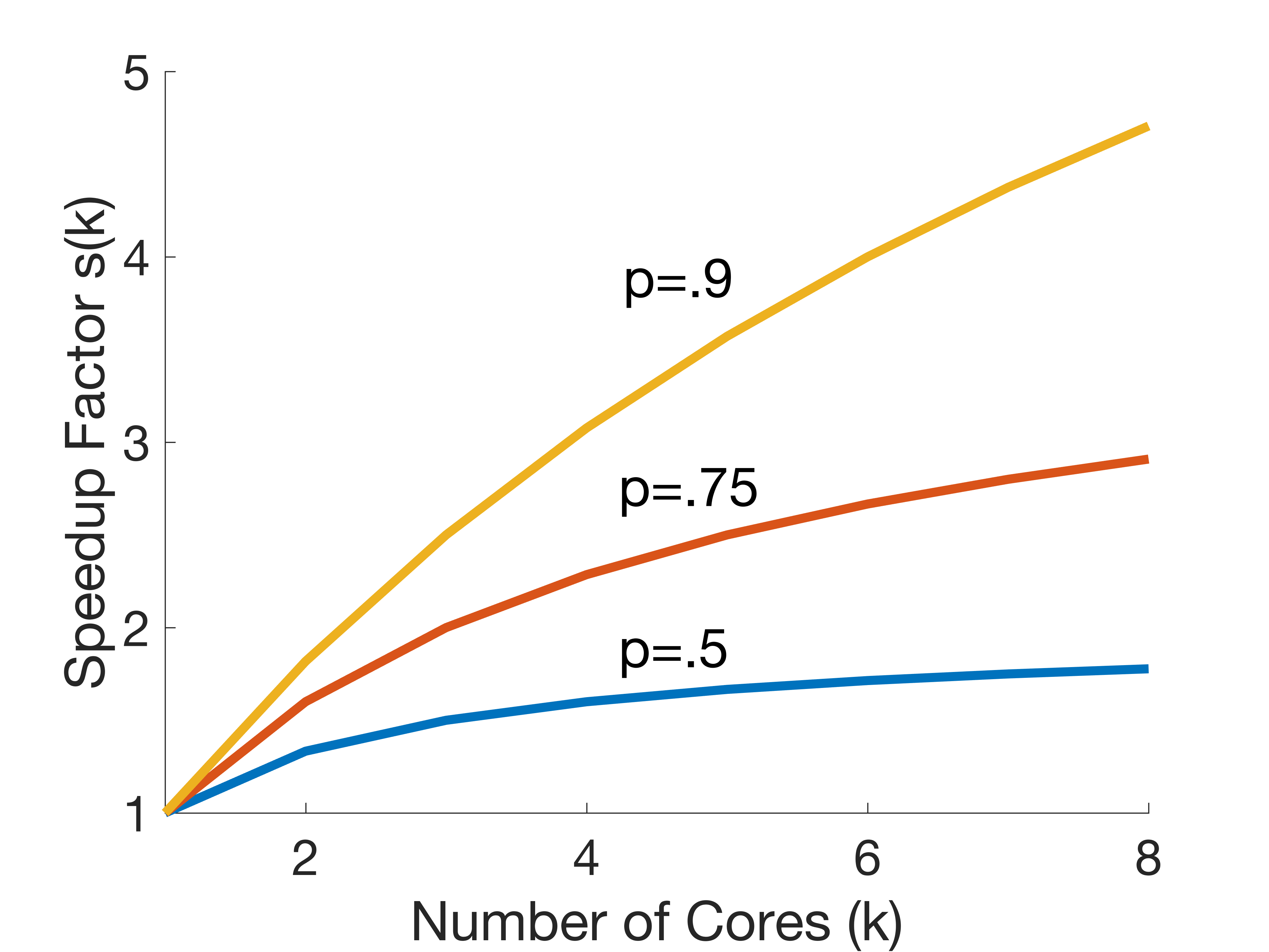}
\label{figAmdahl}
}
\subfigure[]{
\includegraphics[width=.4\textwidth]{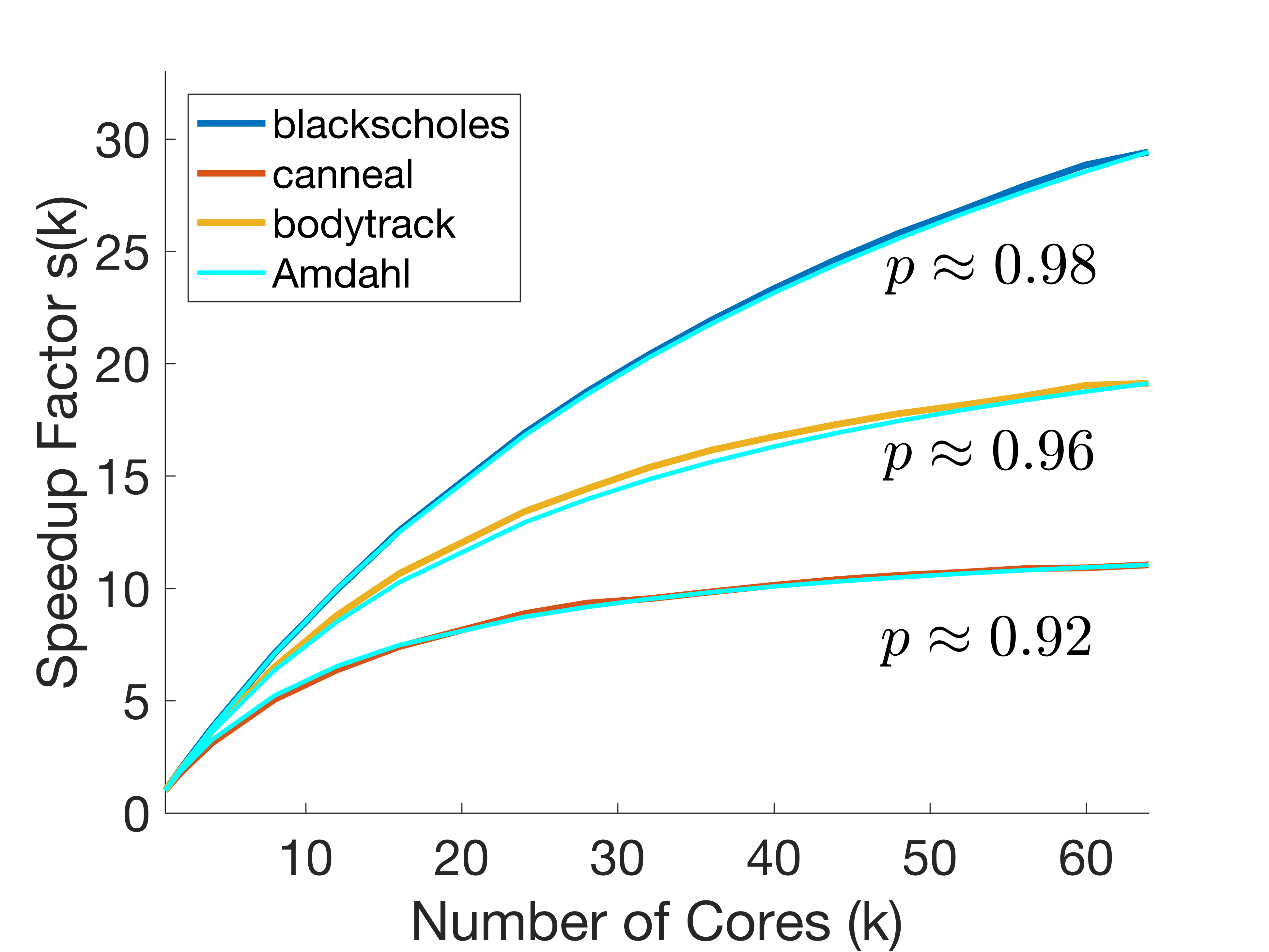}
\label{figParsec}
}
\caption{Various speedup curves under (a) Amdahl's law and (b) the PARSEC-3 Benchmark.  We see that the PARSEC-3 speedup curves are accurately approximated by Amdahl's law; these approximations are shown in (b) in light blue.}

\end{figure}

We mostly assume that jobs are
homogeneous with respect to $s$ -- all jobs receive the same speedup
due to parallelization -- although we will also consider
multiple speedup functions.

As jobs arrive, each job is initially dispatched to some subset of the
cores.  A \emph{scheduling policy} defines, at every moment in time, which
particular cores each job runs on.  Cores are assumed to be
homogeneous, and thus any job is capable of running on any core.

At any moment in time, there may be several jobs (job pieces)
running on a particular core.  We assume each core time-shares between its jobs (i.e. \emph{Processor Sharing}), which is typical for processors in multicore machines.  For example, suppose there are only 2 jobs in the system, each of size $x$, and each runs on the \emph{same} 5 cores.  Each job receives a speedup of $s(5)$ since it runs on 5 cores, but its runtime is $2\cdot\frac{x}{s(5)}$ because each job only gets half of each core.

\subsubsection*{The Problem} 

The fact that modern parallel jobs are capable of running on any number of cores begs the question of how many cores should be used for each job.  Let $T$ denote the \emph{response time} of a job (the time between when the job arrives and when all of its pieces have completed).  Our goal is to find and analyze scheduling policies that aim to minimize the mean response time,
$\mathbb{E}[T]$, across all jobs.  Clearly, $\mathbb{E}[T]$ depends on
the scheduling policy, the speedup function, $s$, the arrival rate of
jobs into the system, $\Lambda$, the mean job size, $\mathbb{E}[X]$,
and the number of cores, $n$.  We define the average system load,
$\rho$, to be 
$$\rho := \frac{\Lambda \mathbb{E}[X]}{n}.$$
 This is equivalently the fraction of time a core is busy when each job is run on a \emph{single} core.  Thus, $\rho < 1$ is a necessary condition for stability, regardless of the scheduling policy used.  We assume both $\rho$ and $s$ are constant over time.

\vspace{-.1in}
\subsection*{Contributions}

\subsubsection*{Contribution 1: Optimal Parallel Scheduling Under a Single Speedup Function}
\vspace{-.05in}

We provide a proof that EQUI is the optimal parallel scheduling policy (assuming exponentially distributed jobs sizes).  EQUI is a policy which first
appeared in \cite{Edmonds1999SchedulingIT}.  Under EQUI, at all times,
the $n$ cores are equally divided among the jobs in the system.
Specifically, whenever there are $\ell$ jobs in the system, then each
job is parallelized across $n/\ell$ cores.  EQUI is an idealized
policy in that (i) EQUI assumes that when a job runs on $n/\ell$ cores its runtime will be distributed as $\frac{X}{s(n/\ell)}$, even if $n/\ell$ is not an integer, and (ii) EQUI
requires jobs to be \emph{malleable} -- every job can change its level
of parallelization while it runs.

\subsubsection*{Contribution 2: Fixed-Width Parallel Scheduling}
\vspace{-.05in}

In practice, jobs are not malleable but are often \emph{moldable} -- a
job can run on any number of cores, $k$, but it must run on the same
$k$ cores for the duration of its lifetime.  In theory, a scheduling
policy could choose a different number of cores on which to run each
incoming job.  We define an extremely simple class of scheduling policies, called \emph{fixed-width} policies, whereby every job is run on the \emph{same fixed
number} of cores.

We take this a step further and impose the constraint that the $n$
cores of the machine are partitioned into static-sized ``chunks'', each
with $k$ cores, where $k$ divides $n$. For example, Figure
\ref{fig:chunks} shows some potential chunkings of a 16-core
machine.  Every arriving job is dispatched to a single chunk and runs
across all the cores in that chunk.

We introduce a policy called \emph{JSQ-Chunk} which uses a
fixed width, $k$, and assigns jobs to chunks according to the
Join-Shortest-Queue dispatching policy.  We provide an approximate
response time analysis for JSQ-Chunk, and, more significantly, we show
how to choose the \emph{optimal fixed-width}, $k^*$, for JSQ-Chunk as a
function of the system load, $\rho$.

\begin{figure}[ht!]
\centering
\subfigure[$k=8$]{
\includegraphics[width=.2\textwidth]{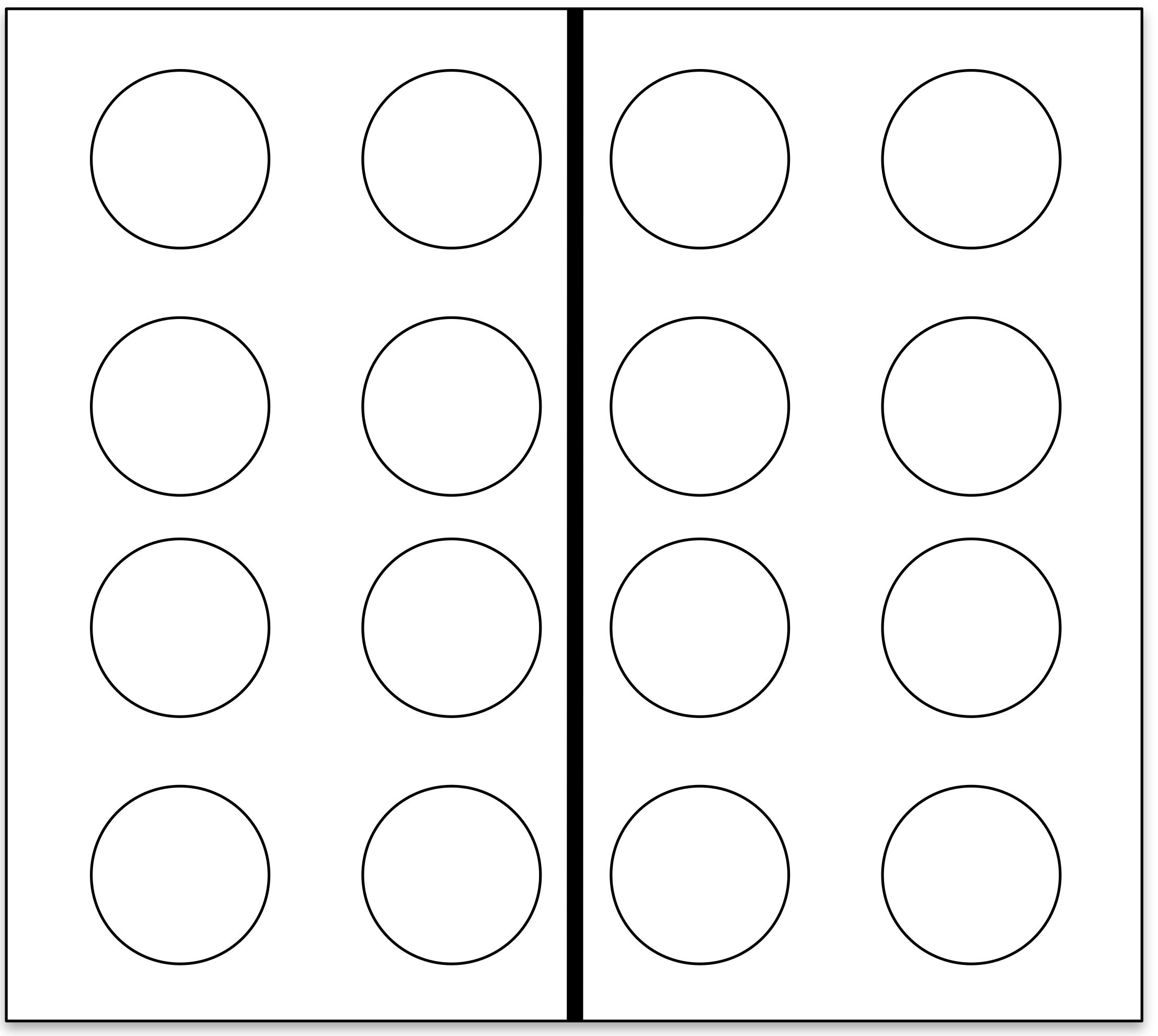}
}
\subfigure[$k=2$]{
\includegraphics[width=.2\textwidth]{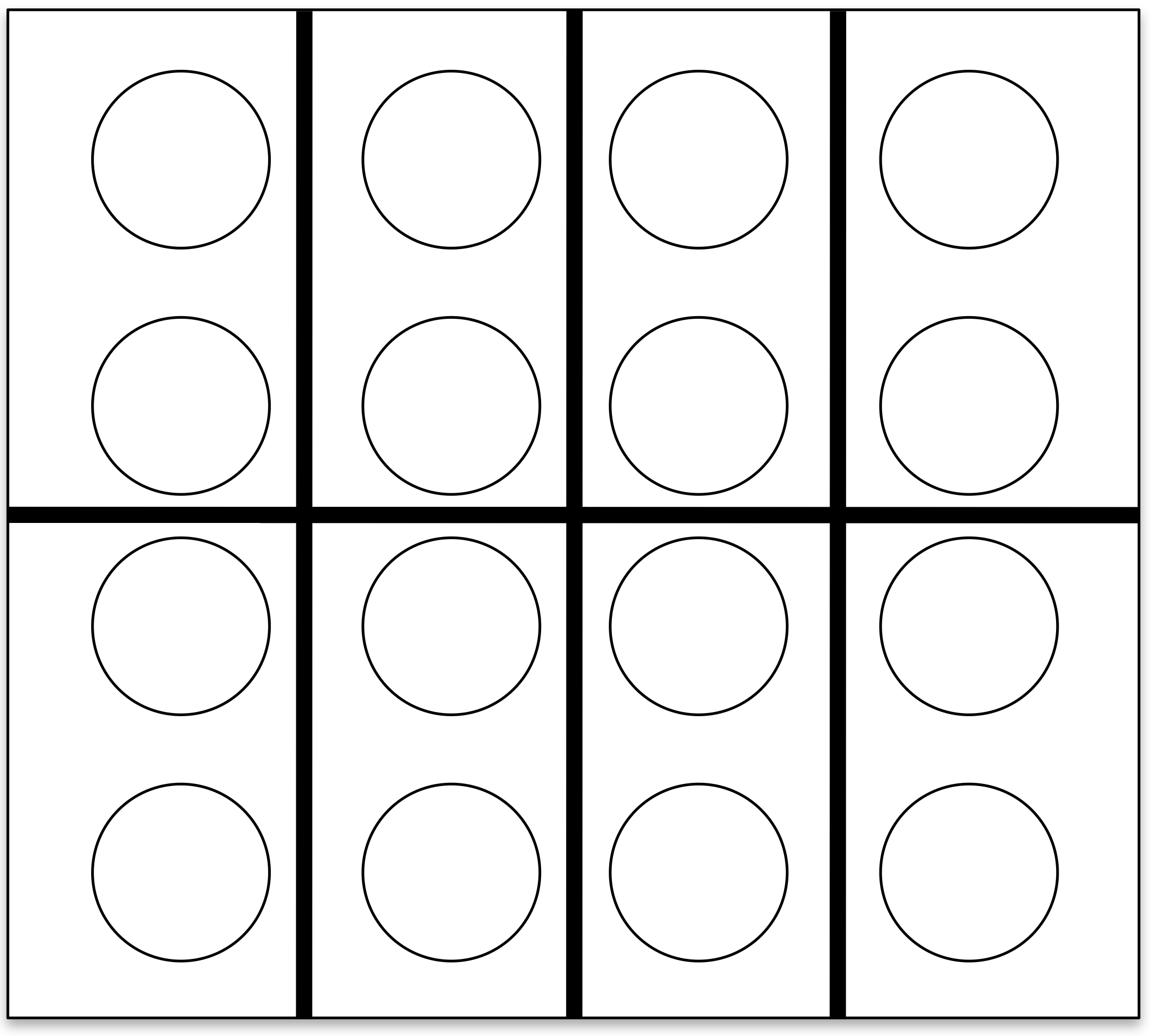}
}
\subfigure[$k=1$]{
\includegraphics[width=.2\textwidth]{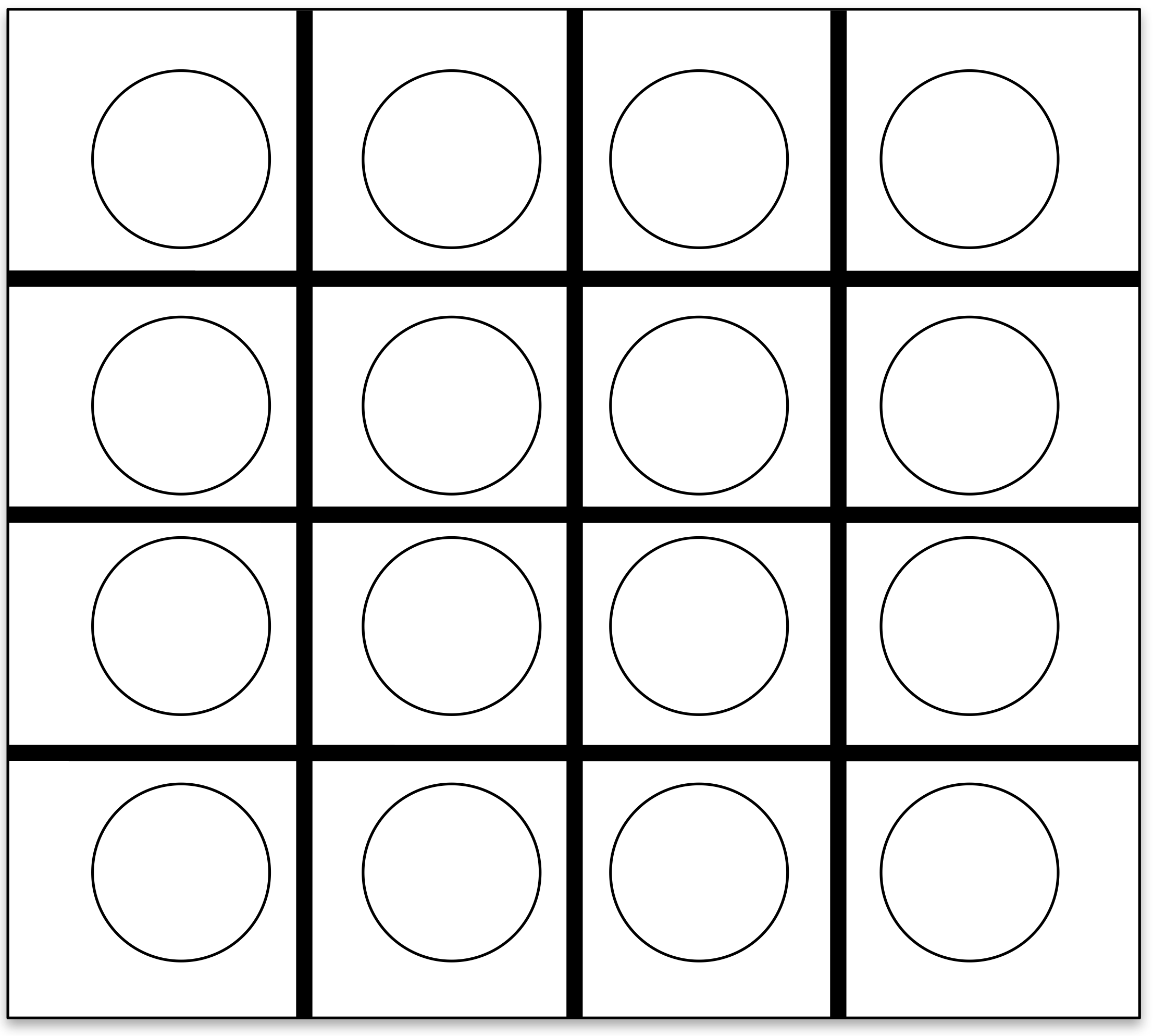}
}
\caption{A machine with $n=16$ cores under various chunk sizes, $k$. Each job is dispatched to a single chunk.}
\label{fig:chunks}
\end{figure}

\subsubsection*{Contribution 3: Near-Optimality of Fixed-Width\\ Scheduling}
\vspace{-.05in}

Fixed-width scheduling policies seem heavily restricted in their
ability to effectively exploit parallelism.  We prove, however, that
certain fixed-width policies, like JSQ-Chunk, can achieve
mean response time converging to that of EQUI as the number of cores,
$n$, becomes large.

We provide a preview of contributions 2 and 3  in Figure \ref{figEQUI}.  Here we see the
mean response times under JSQ-Chunk, as a function of system load $\rho$, for
various choices of $k$.  For a given $\rho$,
this figure shows how to choose the optimal $k^*$.  The figure shows a separate curve for mean response time under each possible choice of $k$ when using JSQ-Chunk.  As expected, increasing $k$ can reduce mean response time, but also reduces the system's stability region.  The shaded regions under these curves denote the values of $\rho$ for which a particular choice of $k$ is optimal.  For example, when $x_8 \leq \rho \leq x_4$, $k^*=4$.  Surprisingly, the performance of JSQ-Chunk, when using the optimal chunk size $k^*$, is not far from that
of EQUI-- we will prove that this difference vanishes in systems with many cores.

\begin{figure}[ht!]
\vspace{-0.1in}
\centering
\includegraphics[width=.5\textwidth]{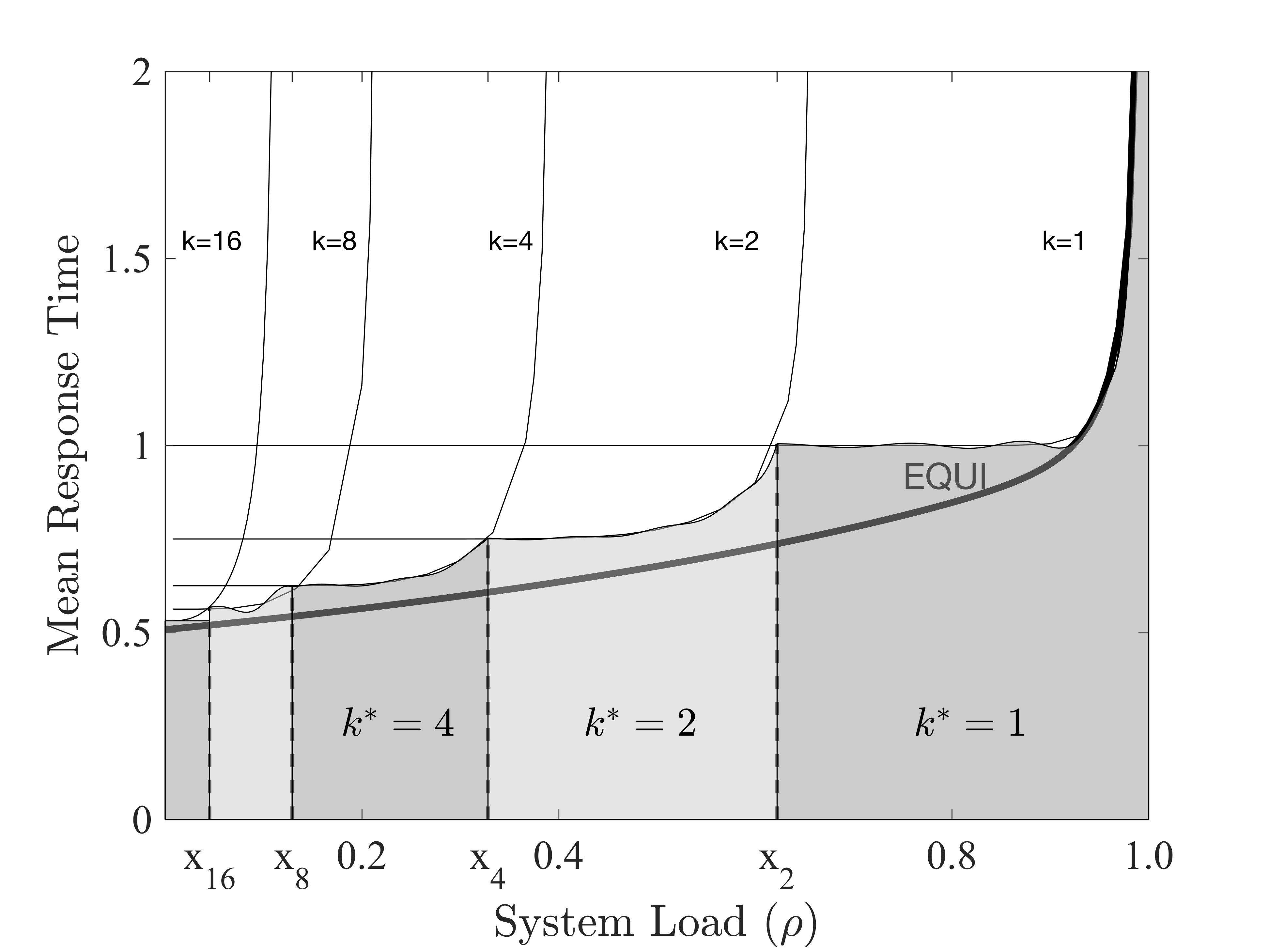}
\vspace{-0.15in}
\caption{Mean response times under JSQ-Chunk (thin line) and EQUI (thick line) when $n=64$.  We assume a speedup curve of Amdahl's law with parameter $p=0.5$ and a shifted Pareto job size distribution with $\alpha=2$ and mean $\mathbb{E}[X]=1$.  
%We see that JSQ-Chunk with the optimal $k^*$ is not far from EQUI.
}
\label{figEQUI}
\vspace{-0.17in}
\end{figure}

\subsubsection*{Contribution 4: Multiple Speedup Functions}
\vspace{-.05in}

Thus far we have assumed that all jobs follow the same speedup function $s$.  We show that, when jobs are permitted to have different speedup functions, EQUI is no longer optimal.  We also show how to numerically compute the optimal policy, OPT, for multiple speedup functions (assuming exponentially distributed job sizes).  
Since finding OPT is computationally intensive we introduce a simple class of policies, called GREEDY, that performs well by maximizing the departure rate.  We prove that one policy in this class, \gstar{}, dominates by both maximizing the overall departure rate and deferring parallelizable work.  We also compare OPT to simple fixed-width policies.
%Since finding the optimal policy may be computationally intensive, we then show how to quickly compute a policy, \gstar{}, which performs well in practice.  Like EQUI, both the optimal policy and \gstar{} assume that jobs are malleable.  We thus also examine the performance of fixed-width policies under multiple speedup functions.  

\section{EQUI: An Optimal Policy}\label{sec-equi}

In order to effectively parallelize jobs, one imagines that using a \emph{dynamic} level of parallelization might greatly improve mean response time.  Specifically, it makes sense to consider policies where the level of parallelization of an incoming job is determined based on the state of the system at the time when the job arrives.  In addition, if jobs are malleable, active jobs could change their levels of parallelization during their lifetimes as the state of the system changes.

EQUI \cite{Edmonds1999SchedulingIT} is a generalization of Processor Sharing to systems with multiple cores.  Under EQUI, whenever there are $\ell$ jobs in the system, each job runs on $n/\ell$ of the cores.
We define $\mu = 1/\mathbb{E}[X]$ to be the rate at which jobs complete when run on a single core.  EQUI is an idealized policy in that it assumes that when a job runs on $n/\ell$ cores its runtime will be distributed as $\frac{X}{s(n/\ell)}$, even if $n/\ell$ is not an integer.  This corresponds to a service rate of $s(n/\ell)\mu$ for each job when jobs are exponentially distributed.  Additionally, we assume that $s(k)=k$ when $k\leq 1$ to account for the effects of Processor Sharing.  Hence, when there are $\ell$ jobs in the system, the total rate at which jobs complete is $\ell\mu s(n/\ell)$ which is equal to $n\mu$ when $\ell \geq n$.  Importantly, EQUI always processes every job in the system at the same rate, even when there are more than $n$ jobs in the system.

Figure \ref{equiChain} shows the Markov chain for EQUI.  This Markov chain assumes that job sizes are exponentially distributed.  However, Lemma \ref{equi-thm-g} proves that EQUI's performance is actually insensitive to the job size distribution.

In Theorem \ref{equi-thm-global-opt} we prove that EQUI is optimal with respect to mean response time when all jobs follow the same speedup function and are exponentially distributed.

\begin{figure}[h!]
\centering
\includegraphics[width=.75\textwidth]{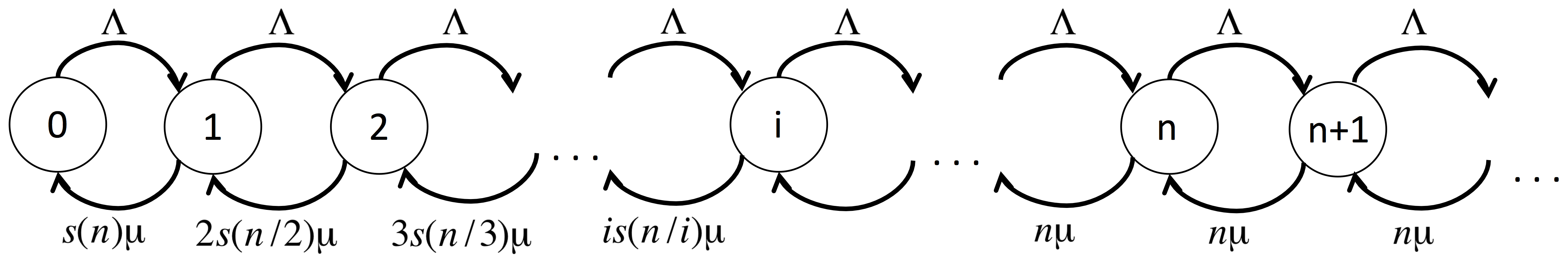}
\caption{Markov chain representing the total number of jobs under EQUI.}
\label{equiChain}
\end{figure}

\subsection{Insensitivity of EQUI}\label{sec:equi:insen}

\begin{lemma}\label{equi-thm-g}
Mean response time under EQUI is insensitive to the job size distribution.
\end{lemma}
\begin{proof}
A system is considered to practice processor sharing if, regardless of how many jobs are in the system, each job receives the same service rate.  Note that under EQUI, when $m$ jobs are in the system, each job receives the same service rate by running on $n/m$ of the cores.  Hence, EQUI fits the definition of a processor sharing system, albeit one with state dependent service rates because the total rate of departures depends on the number of jobs in the system.  It is known from \cite{Baskett:1975:OCM:321879.321887, bonald2002insensitivity} that the mean response time in processor sharing systems with state dependent service rates is insensitive to the job size distribution, thus our performance analysis of EQUI holds for the case where job sizes follow a general distribution.
\end{proof}
\vspace{-.1in}
\subsection{Proving that EQUI is Optimal}
Theorem \ref{equi-thm-global-opt} relies on Lemma \ref{claim-increasing}.
\begin{lemma}\label{claim-increasing}
For any concave, sublinear function, $s$, the function $i\cdot s\left(\frac{n}{i}\right)$ is increasing in $i$ for all $i$ for all $i<n$, and is non-decreasing in $i$ for all $i\geq n$.
\end{lemma}
\begin{proof}
See Appendix \ref{pf-increasing}
\end{proof}
\begin{theorem}\label{equi-thm-global-opt}
Within our model, assuming jobs are malleable, have exponentially distributed sizes, and all follow the same speedup function, $s$, $\mathbb{E}[T]^{\mbox{EQUI}}\leq \mathbb{E}[T]^{\mbox{P}}$ for any scheduling policy $P$.
\end{theorem}
\begin{proof}
Let $P$ be a scheduling policy which processes malleable jobs and currently has $i$ active jobs (the system is in state $i$).  In every state, $i$, $P$ must decide (i) how many jobs, $j$, to run and (ii) how to allocate the $n$ cores amongst the $j$ jobs.  Hence, for some $ \vec{\theta}$, the rate of departures from state $i$ under $P$ is at most
\begin{equation}\mu\sum_{k=1}^j s(n\theta_k)\label{eq-dep}\end{equation}
where $0< j \leq \min(i,n)$, $\theta_k > 0$ for all $1\leq k \leq j$, and $\sum_{k=1}^j \theta_k=1$.  Note that \eqref{eq-dep} is an upper bound on $P$'s departure rate, since it assumes that the $j$ jobs run on disjoint partitions of cores; given that $s$ is concave the total rate of departures will not increase when a single core (or fraction of a core) is allocated to more than one job.

We also know that
$$\frac{1}{j}\sum_{k=1}^j s\left(n\theta_k\right)\leq s\left(\frac{n}{j}\right)$$
again, by the concavity of $s$.  Hence,
$$\sum_{k=1}^j s(n\theta_k) \leq j \cdot s\left(\frac{n}{j}\right) \indent \forall\ 0< j\leq i,\ \forall\ \vec{\theta} \in (0, 1]^j$$
and thus an upper bound on $P$'s total rate of departures from any state is of the form
$$j \cdot s\left(\frac{n}{j}\right)\mu.$$

By Lemma \ref{claim-increasing}, $j\cdot s\left(\frac{n}{j}\right)$ is non-decreasing in $j$.  Thus, an upper bound on $P$'s rate of departures from any state $i$ is 
\begin{equation}i \cdot s\left(\frac{n}{i}\right)\mu.\label{equi-dep-ub}\end{equation}
Furthermore, we can see that EQUI achieves this departure rate in every state, $i$.  Hence, \eqref{equi-dep-ub} is the maximal rate of departures from any state, $i$.

We can now compare the Markov chain corresponding to $P$ with the Markov chain for EQUI (Figure \ref{equiChain}) and relate the number of jobs in the system, $N$, under both policies.  Both chains have the same arrival rates, but the rate of departures under EQUI is greater than or equal to the rate of departures under $P$ in any state, $i$, since the departure rate chosen by EQUI is maximal.  Thus,
$$\mathbb{E}[N]^{\mbox{EQUI}}\leq \mathbb{E}[N]^{\mbox{P}}$$
and, by Little's Law,
$$\mathbb{E}[T]^{\mbox{EQUI}}\leq \mathbb{E}[T]^{\mbox{P}}.$$

\end{proof}
\vspace{-.1in}
\subsection{EQUI with General Size Distributions}\label{sec:opt:general}

Theorem \ref{equi-thm-global-opt} requires exponentially distributed job sizes.  One might think that EQUI's insensitivity (Lemma \ref{equi-thm-g}) might imply that EQUI is optimal under all job size distributions.  However, this is false.  Consider for example the case where jobs follow a Pareto distribution with decreasing hazard rate.  Independent of the speedup function, the optimal policy should devote more cores to jobs which have lower ages, and are thus more likely to finish in the immediate future.  This differs from EQUI's equal division of resources.

In general, optimal scheduling of jobs with general job size distributions is an extremely difficult problem.  Even optimally scheduling generally distributed jobs on a \emph{single core} requires using a policy derived from the Gittins index \cite{ziv}.  In our case, where the optimal scheduling policy must allocate many cores, finding the optimal policy is at least as difficult.  In fact, it is not known how to even numerically compute the optimal policy in this case.  We therefore defer this work to a future paper. 

\section{Fixed-Width Policies}\label{rand}
While EQUI performs very well, it requires jobs to change their levels of parallelization while running.  Fixed-width policies require that jobs be moldable, but not necessarily malleable, which is more realistic for certain workloads.  In general, a fixed-width policy is any policy which chooses a fixed level of parallelization, $k$, to use for each arriving job.  Every arrival is parallelized across $k$ cores, and is run on the same $k$ cores for its entire lifetime.  In this section, we consider several natural fixed-width policies.

We begin by defining the Random dispatching policy, which parallelizes each arriving job across $k$ cores chosen uniformly at random.  It turns out (see Theorem \ref{rc-vs-r}) that Random is dominated by the Random-Chunk policy; hence we devote Section \ref{rc} to Random-Chunk.  In Section \ref{jsq}, we consider an improved fixed-width policy called JSQ-Chunk, which dispatches arrivals to the chunk with the shortest queues.

Like EQUI, Random-Chunk and JSQ-Chunk are insensitive to the job size distribution.  Thus, our analysis of these policies admits general job size distributions, $X$.
\vspace{-.1in}
\subsection{Random-Chunk}\label{rc}
The motivation behind the Random-Chunk policy is that we would like to use random assignment while ensuring that the different pieces of a single job complete at the same time.  For a level of parallelization, $k$, we begin by partitioning the cores into $c=n/k$ chunks of size $k$.
We will only consider values of $k$ which divide $n$, creating a uniform partition of the cores.  When a job arrives, a chunk is selected uniformly at random, and the job is parallelized across all $k$ cores in this chunk.

It is easy to see that under Random-Chunk each piece of a job will at all times experience the same state, and hence each of the $k$ job pieces will complete at the same time.  Consequently, the response time for a job is equal to the response time of any of its $k$ pieces, which in turn is simply the mean response time at a single core.  This allows us to easily derive an expression for the overall mean response time under Random-Chunk with a level of parallelization $k$ in Theorem \ref{rc-thm-mrt} below.  

\begin{theorem}\label{rc-thm-mrt}
The mean response time under Random-Chunk with a level of parallelization, $k$, is given by
\begin{equation*}
	\mathbb{E}[T]^{\mbox{\scriptsize Rand-Chunk}}=\frac{\mathbb{E}[X]}{s(k) - k\rho}.\label{rc-ps}
\end{equation*}
\end{theorem}
\begin{proof}
Recall that cores follow the PS scheduling discipline.  We know from \cite{kleinrock1975queueing2} that the mean response time in a single $M/G/1/PS$ queue with arrival rate $\Lambda$, mean job size $\mathbb{E}[X]$ and load $\rho=\frac{\Lambda \mathbb{E}[X]}{n}$ is given by
$$\mathbb{E}[T] = \frac{\mathbb{E}[X]}{1-\rho}.$$

It will suffice to analyze the response time of one core in our system.  Under a level of parallelization of $k$, the arrivals into a single core follow a Poisson process with rate 
$$\Lambda_k:=\frac{\Lambda}{c}$$
 and the mean service requirement of job pieces is $\mathbb{E}[X_k]$ (see \eqref{eq:service}).  Thus, we can define the load on a single core in this system to be 
 $$\rho_k := \Lambda_k \mathbb{E}[X_k].$$
 This allows us to express $\mathbb{E}[T]^{\mbox{\scriptsize Rand-Chunk}}$  under a level of parallelization of $k$ as:
$$\mathbb{E}[T]^{\mbox{\scriptsize Rand-Chunk}} = \frac{\mathbb{E}[X_k]}{1-\rho_k}=\frac{\mathbb{E}[X_k]}{1-\rho\frac{k}{s(k)}}=\frac{\mathbb{E}[X]}{s(k) - k\rho}.$$
\end{proof}

\begin{corollary} \label{cor-k-str} The optimal Random-Chunk policy uses chunk size $k^*$, where 
$$k^*=\argmin_k \frac{\mathbb{E}[X]}{s(k) - k\rho}.$$
If we define the vector $\pmb{k}=(k_1,k_2,\ldots k_m)$ to be the factors of $n$ in increasing order, this implies:

\begin{equation*}k^* = \begin{cases}
k_m & 0 \leq \rho \leq \frac{s(k_{m}) - s(k_{j-1})}{\mathbb{E}[X]}\\
k_i & \frac{s(k_{i+1}) - s(k_{i})}{\mathbb{E}[X]} < \rho \leq  \frac{s(k_{i}) - s(k_{i-1})}{\mathbb{E}[X]}, 1 < i < m\\
k_1 & \frac{s(k_{2}) - s(k_{1})}{\mathbb{E}[X]} < \rho.\\
\end{cases}
\end{equation*}
\end{corollary}
The results of this analysis are shown in Figure \ref{fig-rc-fcfs}.  This figure shows the mean response time under various choices of $k$, and the shaded regions below the curves denote the values of $\rho$ for which each choice of $k$ is optimal.  Under high load, we see that the system is unable to tolerate wasting resources by parallelizing jobs, and thus each job is run on a single core ($k^*=1$).  Conversely, under light load it is beneficial to pay this cost of parallelization in order to reduce the service requirement of individual jobs.
 \begin{figure}[h!t]
\centering
\subfigure{
\includegraphics[width=.6\textwidth]{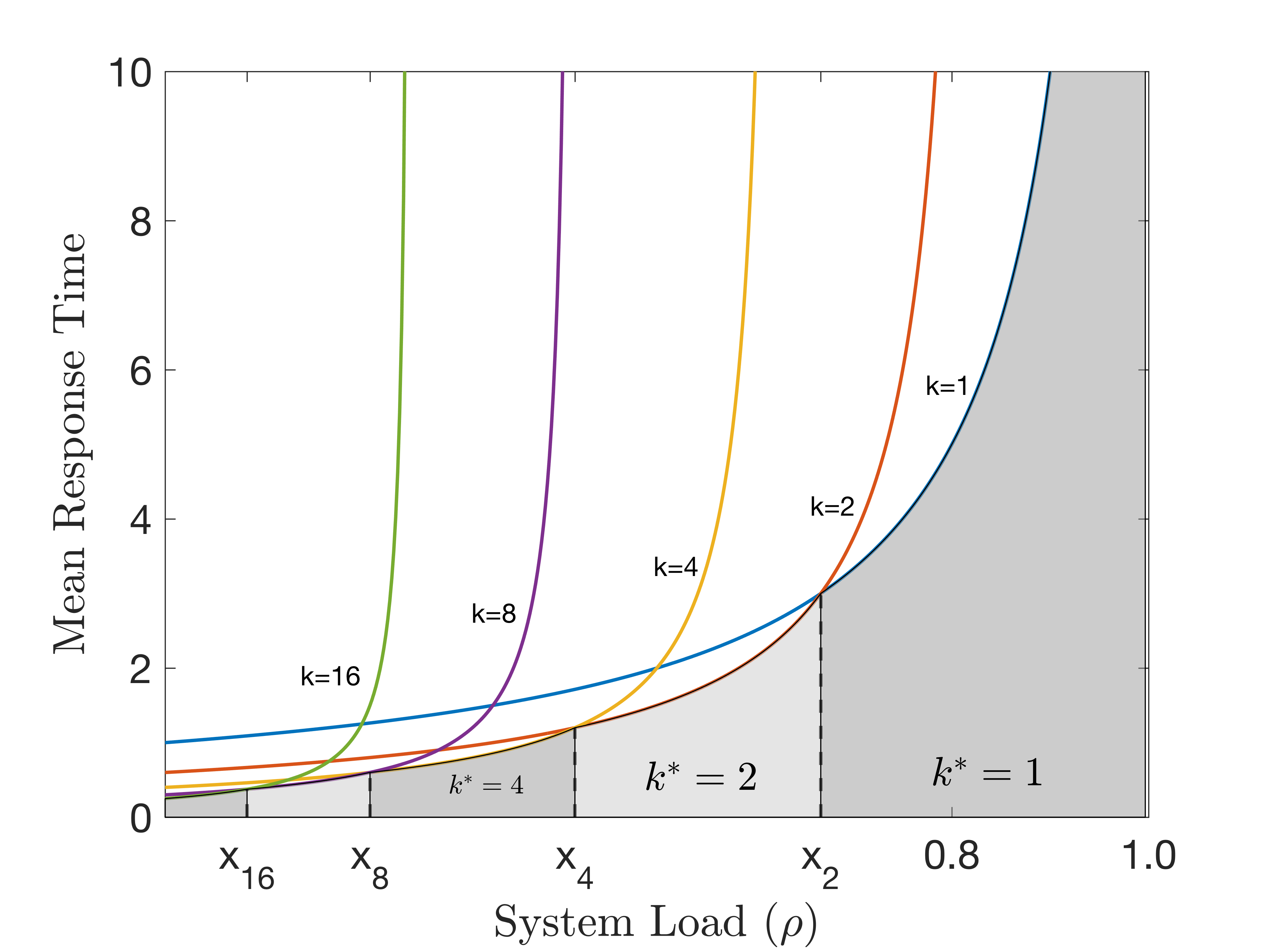}
}
\caption{Results from analysis.  Mean response times under Random-Chunk dispatching, a system with $n=16$ cores, various choices of the level of parallelization $k$, and a hyperexponential job size distribution with $\mathbb{E}[X]=1$ and $C^2=10$.  The speedup curve used is Amdahl's law with parameter $p=0.8$.}
\label{fig-rc-fcfs}
\end{figure} 

\subsection*{Random-Chunk vs. Random}
As we have seen, Random-Chunk yields much easier analysis than Random.  Fortunately, Theorem \ref{rc-vs-r} tells us that Random-Chunk results in lower mean response times than Random, rendering the analysis of Random irrelevant.
\begin{theorem}
The mean response time under Random-Chunk with a fixed level of parallelization, $k$, is less than under Random with the same $k$.\label{rc-vs-r}
\end{theorem}
\begin{proof}
Regardless of whether we are using Random or Random-Chunk dispatching, the arrival process into a given core is Poisson with rate $\Lambda_k$ and the service requirement of job pieces is distributed as $X_k$.  Thus, the response time for the $i$th job \emph{piece} is distributed as some random variable $T_i$, and the response times of all pieces are identically distributed in both cases.

Under Random-Chunk, the response times of a job's pieces are identical.  Hence, under Random-Chunk, the response time of a \emph{job} is simply the response time of one of its pieces.  Under Random, the response time of a job is the \emph{maximum} of the $k$ response times of its $k$ pieces.  Thus, for a single job, the response time under Random-Chunk has distribution $T_1$ while under Random the response time of a single job has distribution 
$$\max(T_1, T_2, \ldots , T_k)$$ 
where $T_i \sim T$.  Since $T \leq_{st} \max(T_1, T_2, \ldots, T_k)$,
 we see that, for a given level of parallelization, Random-Chunk's mean response time is less than that of Random. 
\end{proof}
\vspace{-.1in}
\subsection*{Mixed-Random-Chunk}\label{mrc}
Until now, we have studied Random-Chunk policies where each chunk consists of $k$ cores.  It is conceivable, however, that allowing for type 1 chunks of size $k_1$ and type 2 chunks of size $k_2$ could improve mean response time.  

Having two chunk sizes introduces more dispatching parameters.  One must decide how many cores, $a$, to devote to type 1 chunks.  One also must decide what fraction, $p$, of arrivals to assign to type 1 chunks.  Finding the optimal Mixed-Random-Chunk policy then requires simultaneously optimizing over not just $k_1$ and $k_2$, but additionally $a$ and $p$.
 
We find that while some instances do exist where a Mixed-Random-Chunk policy is better than Random-Chunk, these instances occur very sparsely throughout the parameter space ($k_1, k_2, a,p$), and the advantage gained is very slight at best.  In fact, if we limit ourselves to cases where the arrival rate into each chunk is proportional to the chunk's size, we can prove that Mixed-Random-Chunk is never advantageous in terms of mean response time or variance of response time (see Appendix \ref{sec-mrc}). 
\vspace{-.1in}
\subsection{JSQ-Chunk}\label{jsq}

The Random-Chunk policy can be greatly improved by considering the state of the system when dispatching jobs.  Under JSQ-Chunk with a level of parallelization, $k$, we again partition the system into chunks of size $k$.  When a job arrives to the system, it is parallelized across the $k$ cores in the chunk that is currently serving the fewest jobs.  We expect that JSQ-Chunk will have much lower mean response times than Random-Chunk and will admit higher values of $k^*$.  To determine the correct $k^*$ for JSQ-Chunk, it is important that we can accurately analyze mean response time under JSQ-Chunk. 

JSQ-Chunk is based on JSQ, an old and well-studied dispatching policy in the traditional queueing literature which assumes no parallelization.  Under JSQ, the system consists of $c$ queues and every incoming arrival is immediately dispatched to the queue with the smallest number of jobs.  While analyzing response times in a traditional (non-parallel) system with JSQ dispatching is notoriously hard, \cite{Nelson1993PerfromEval} provides a good closed-form approximation of mean response time.  The approximation in \cite{Nelson1993PerfromEval} assumes FCFS queues with exponentially distributed job sizes.  However, as shown in \cite{gupta2007analysis}, the case of PS queues with generally distributed job sizes has nearly the same mean response time as the case of FCFS queues with exponentially distributed service requirements.  Thus, the approximation in \cite{Nelson1993PerfromEval} still applies to our model of PS queues and generally distributed job sizes.  This approximation states:
$$\mathbb{E}[T]^{\mbox{JSQ}}(\Lambda, c, \mathbb{E}[X])\approx W_{M/M/c}(\rho)S(\rho)R(\rho) + \mathbb{E}[X],$$
where $\Lambda$ is the arrival rate, $c$ is the number of queues, $X$ is a random variable representing the size of a job, and $\rho=\Lambda\mathbb{E}[X]/c$.  $W_{M/M/c}$ is the mean time in queue in an $M/M/c$ queueing system, and $S(\rho)$ and $R(\rho)$ are experimentally derived correction factors, given in Appendix \ref{app-np}.

\begin{observation}\label{obsJSQ}
Our parallel system with JSQ-Chunk dispatching with level of parallelization $k$, $n$ cores, $c=n/k$ chunks, total arrival rate $\Lambda$, and job size distribution $X$ has the same mean response time as a traditional JSQ queueing system with, $c$ queues, a total arrival rate of $\Lambda$, and job size distribution $X_k=\frac{X}{s(k)}$.
\end{observation}
Observation \ref{obsJSQ} follows from the fact that, under JSQ-Chunk, all the cores within a chunk have the same queue state and receive the same arrivals.  This allows us to directly apply the approximation in \cite{Nelson1993PerfromEval} to analyze JSQ-Chunk as follows:
\begin{align}
\mathbb{E}[T]^{\mbox{JSQ-Chunk}}(\Lambda, n, \mathbb{E}[X], k) &= \mathbb{E}[T]^{\mbox{JSQ}}(\Lambda, n/k, \mathbb{E}[X_k]).
\label{eq-np-approx}
\end{align}
\begin{figure*}[h!t]
\centering
\subfigure[JSQ-Chunk $n=16$]{
\includegraphics[width=.45\textwidth]{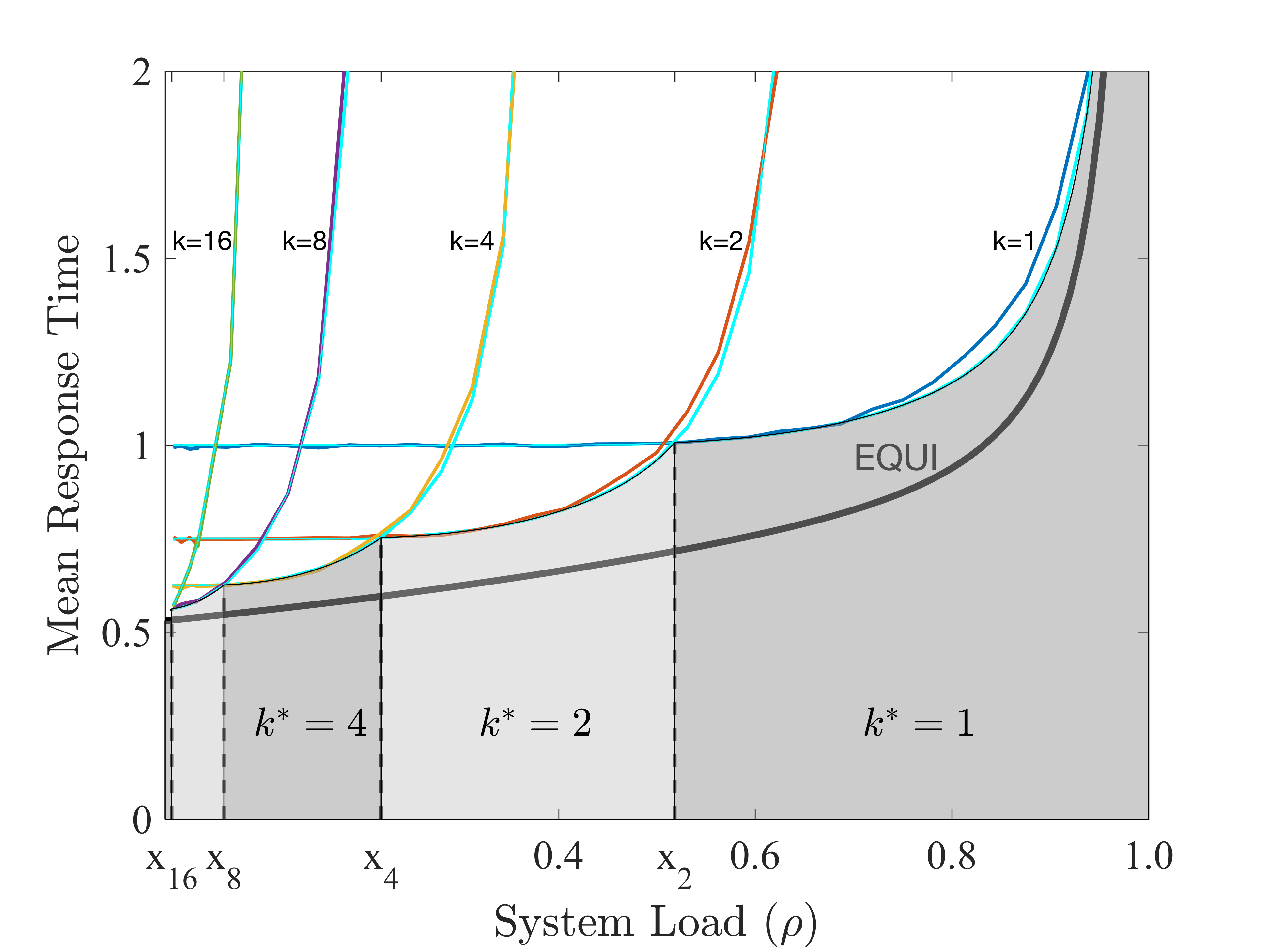}
}
\subfigure[Random-Chunk $n=16$]{
\includegraphics[width=.45\textwidth]{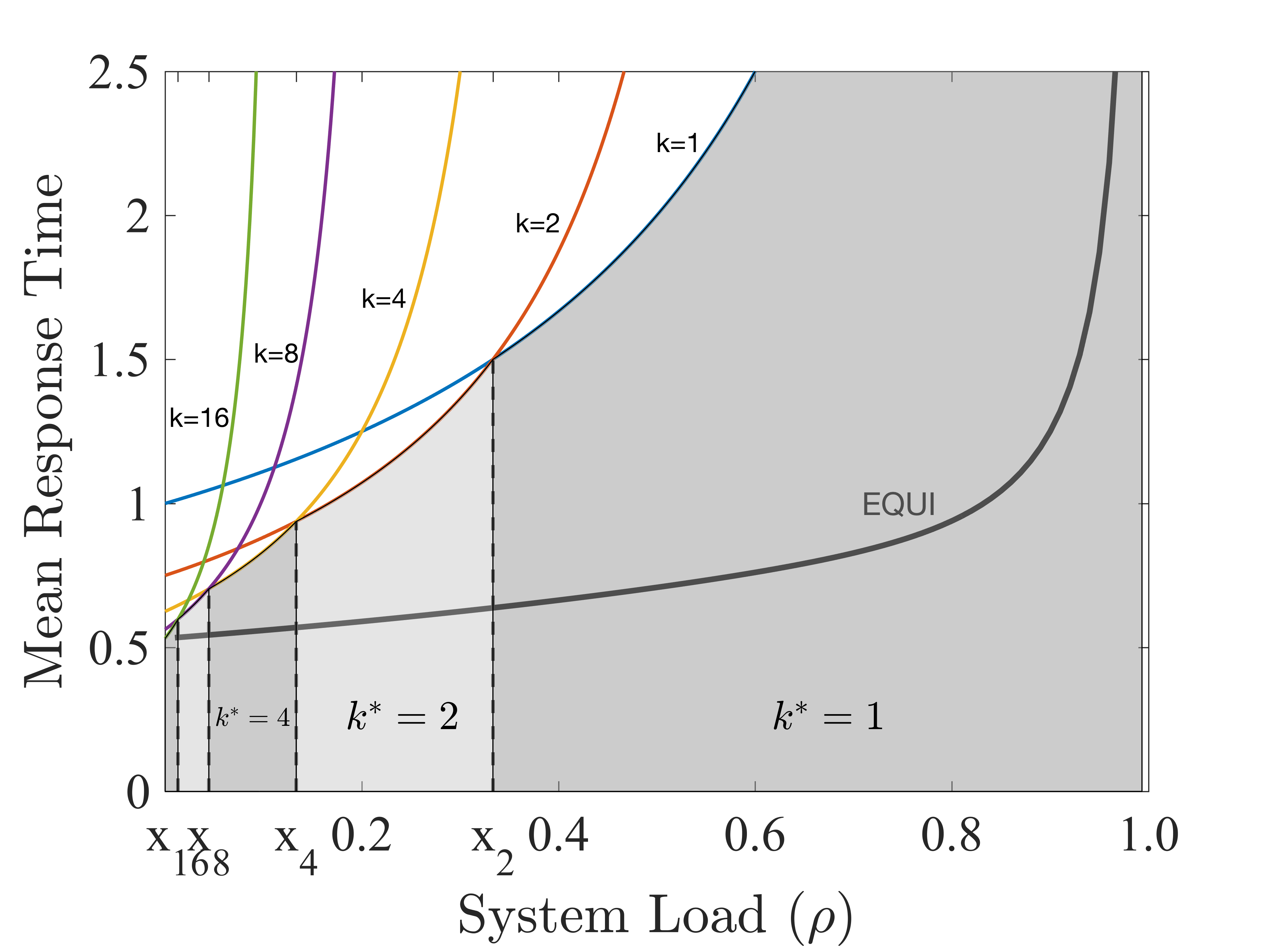}
}
\caption{Response times under (a) JSQ-Chunk dispatching and (b) Random-Chunk dispatching.  In (a), results are show from both analysis and simulation (the jagged lines), which largely overlap.  Both graphs also show the results of analysis of EQUI.  We assume a speedup curve of Amdahl's law with a parameter of $p=0.5$ and exponentially distributed job sizes with mean $\mathbb{E}[X]=1$.}
\label{figJSQ}
\end{figure*}
\section{JSQ-Chunk Converges to EQUI}

\subsection{The Performance of JSQ-Chunk}

Figure \ref{figJSQ} shows the mean response time under JSQ-Chunk (left graph) as a function of system load $\rho$, as computed using the approximation in \eqref{eq-np-approx}.  In addition, we also show results from simulation which lie almost on top of the approximation results.  Fortunately, we see that the approximation for JSQ-Chunk suffices to accurately derive $k^*$ values.  Analogous results of the analysis of Random-Chunk are shown in the right graph.  We find that, for a given load $\rho$, $k^*$ is generally lower under Random-Chunk dispatching as compared with JSQ-Chunk due to JSQ-Chunk's  superior load balancing.  Furthermore, we see that the performance of JSQ-Chunk is much closer to that of EQUI than Random-Chunk.  While the response time of Random-Chunk does not depend on the total number of cores, it is not clear how JSQ-Chunk will behave as the number of cores becomes large.  We will now see that mean response time under JSQ-Chunk approaches that of EQUI as the system scales.

\begin{figure}[h!]
\centering
\includegraphics[width=.5\textwidth]{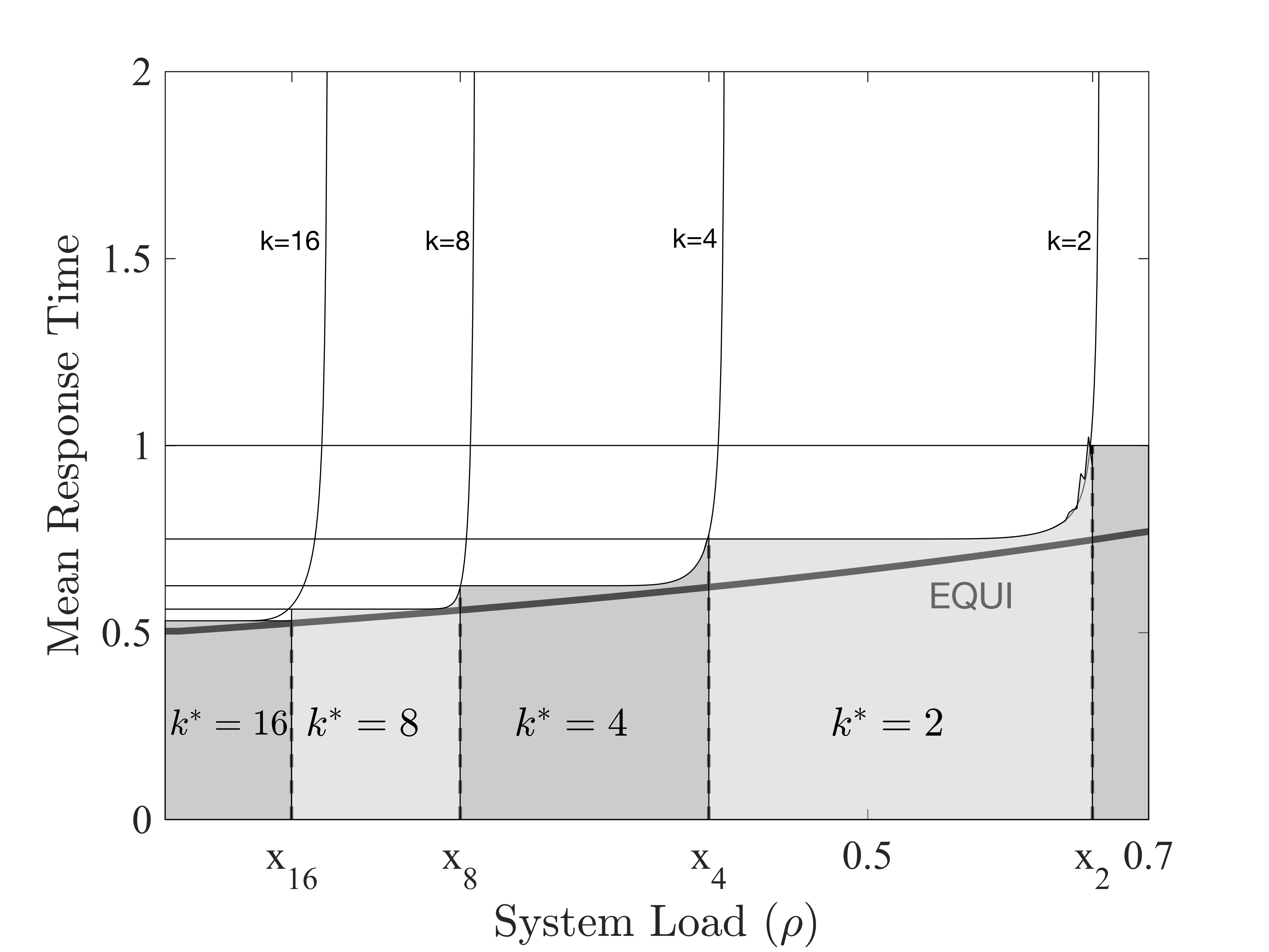}
\caption{Analysis of mean response time under JSQ-Chunk dispatching and EQUI with a large number of cores ($n=512$).  We assume a speedup curve of Amdahl's law with a parameter of $p=0.5$ mean job size $\mathbb{E}[X]=1$.  Increasing the number of cores narrows the gap between JSQ-Chunk and EQUI, especially at the labeled critical load points.}
\label{figClose}
\end{figure}

\subsection{Why JSQ-Chunk is Close to EQUI}\label{jsq:near:opt}

As we saw in Figure \ref{figJSQ}, the mean response time under JSQ-Chunk (with optimal $k^*$) is close to that under EQUI.  Figure \ref{figClose} shows that, as we increase the number of cores to $n=512$, JSQ-Chunk becomes even closer to EQUI.  This is surprising because JSQ-Chunk uses a fixed level of parallelization, while EQUI continuously changes its level of parallelization based on the system state.  This phenomenon can be viewed as  EQUI choosing an \emph{effective} level of parallelization of $k^*$.  Figure \ref{figClose} also illustrates the \emph{critical load points}, $x_{k^*}$ for values of $k^*$, at which JSQ-Chunk is indifferent between two choices of $k^*$.  As $n$ increases, these critical load points converge to the instability points for the corresponding choices of $k^*$, as the curves become more L-shaped.  In looking at Figure \ref{figClose}, we see that it is just below these critical load points that JSQ-Chunk's performance is closest to the performance of EQUI.  The rest of this section is devoted to formally stating and proving this observation (see Theorem \ref{thm:jsq:opt} below).

We first note that, given a fixed value of $\lambda$, a fixed mean job size $\mathbb{E}[X]$, and a speedup function $s$, the choice of the optimal level of parallelization under JSQ-Chunk, $k^*$, depends only on $n$, the number of cores in the system.  Let $k^*(n)$ denote this optimal level of parallelization given some values of these other system parameters.  In particular, note that $\Lambda = \lambda n$ will increase with $n$, and thus the system load, $\rho$, remains fixed for all values of $n$.  When $n$ is small, changes in $n$ will have a significant impact on the values of $k^*(n)$, but for sufficiently large values of $n$, $k^*(n)$ will become constant in $n$.  This is summarized in the following lemma regarding the performance of JSQ-Chunk.

\begin{lemma}\label{lemma:jsq}
Given a fixed value of $\lambda$, a fixed value of $\mu = \frac{1}{\mathbb{E}[X]}$, and some speedup function $s$, let $k^*(n)$ denote the optimal level of parallelization under JSQ-Chunk in a system of size $n$.  Let $\mathbb{E}[T]^{\mbox{JSQ-Chunk}}$ be the mean response time under JSQ-Chunk with level of parallelization $k^*(n)$.  There exists some constant $k^*$, not dependent on $n$, such that 
$$\lim_{n\rightarrow \infty} k^*(n) = k^*$$
and
$$\lim_{n\rightarrow \infty} \mathbb{E}[T]^{\mbox{JSQ-Chunk}} = \frac{1}{s(k^*)\mu}.$$
\end{lemma}
\begin{proof}
We can observe that, in this case, the system load $\rho = \frac{\lambda}{\mu}$ is constant (does not depend on $n$).  We know from \cite{harchol2013performance} that for a fixed system load $\rho$, the probability of queuing under JSQ (and hence also JSQ-Chunk) vanishes as the number of cores becomes large.  Thus, the mean response time under JSQ-Chunk with any level of parallelization, $k$, such that the system is stable, will converge to $\frac{1}{s(k)\mu}$ where $\mu = \frac{1}{\mathbb{E}[X]}$.  Since load is fixed, there exists some level of parallelization, $k^*$, which is the highest value of $k$ for which the JSQ-Chunk system is stable.  We know that the mean response time under JSQ-Chunk with level of parallelization $k^*$ converges to $\frac{1}{s(k^*)\mu}$.  For any $k>k^*$, the system is unstable.  For any $k<k^*$, the mean response time under JSQ-Chunk with level of parallelization $k$ converges to $\frac{1}{s(k)\mu} \geq \frac{1}{s(k^*)\mu}$ since the speedup function is non-decreasing.  Thus,
$$\lim_{n\rightarrow \infty} \mathbb{E}[T]^{\mbox{JSQ-Chunk}} = \frac{1}{s(k^*)\mu}$$
as desired.
\end{proof}

We will refer to $\lim_{n\rightarrow \infty} k^*(n)$ as $k^*$ for the remainder of the section.

We now want to show that as $n \rightarrow \infty$, $\mathbb{E}[T]^{\mbox{EQUI}} \rightarrow \frac{1}{s(k^*)\mu}$ also.  That is, \emph{EQUI is behaving as if it was using a fixed level of parallelization of $k^*$}.

To relate the performance of EQUI to $k^*$, it will be helpful to examine the behavior of EQUI at the critical load points.  Recall that we can define $$\Lambda_{k^*} =  \frac{\Lambda}{c}= \frac{\lambda n}{c},$$
where $c=\frac{n}{k^*}$ is the number of chunks and 
$$\rho_{k^*} = \Lambda_{k^*}\mathbb{E}[X_{k^*}]$$
is the load observed by a chunk of size $k^*$.  We can see from Lemma \ref{lemma:jsq} that the critical load points move towards the instability points as $n$ increases.  Thus, we define a critical load point under a level of parallelization of $k^*$ to be the point where $\rho_{k^*}=1$.  Observe that 
$$\rho_{k^*} =1 \Longleftrightarrow \Lambda = c s(k^*) \mu.$$
We now evaluate EQUI at critical load points where $\Lambda= c s(k^*) \mu$.  Note that when $k^*=1$, the mean response time under both EQUI and JSQ-Chunk will tend towards infinity.  Thus, we only consider cases where $k^* > 1$.

\begin{theorem}\label{thm:jsq:opt}
Given any fixed $k^*>1$, let $\mathbb{E}[T]^{\mbox{EQUI}}_{\rho_{k^*}=1}$ be the mean response time under EQUI when $\rho_{k^*} = 1$.  Then,
\begin{equation}\lim_{n\rightarrow \infty} \mathbb{E}[T]^{\mbox{EQUI}}_{\rho_{k^*}=1} =\frac{1}{s(k^*) \mu}=\lim_{\rho_{k^*}\rightarrow 1^-}\lim_{n\rightarrow \infty} \mathbb{E}[T]^{\mbox{JSQ-Chunk}}.\label{eq:jsq:opt}\end{equation}
\end{theorem}

\begin{proof}
In this proof we limit our discussion to exponentially distributed job sizes.  However, because EQUI and JSQ-Chunk are insensitive to the job size distribution (see Section \ref{sec:equi:insen} and Section \ref{jsq}), our proof generalizes to the case where jobs are generally distributed.  We can see that the right hand side of this claim follows directly from Lemma \ref{lemma:jsq}, and we thus proceed to analyzing the mean response time under EQUI.

It will be helpful to start by examining a \emph{threshold chain} as shown in Figure \ref{thresholdChain}.  Observe that the service rate below the threshold state is $\mu_{low}$ and the service rate above the threshold state is $\mu_{high}$.  We define
$$\rho_{low} := \frac{\Lambda}{\mu_{low}}, \qquad \rho_{high} := \frac{\Lambda}{\mu_{high}}.$$

\begin{figure}[b!]
\centering
\includegraphics[width=.75\textwidth]{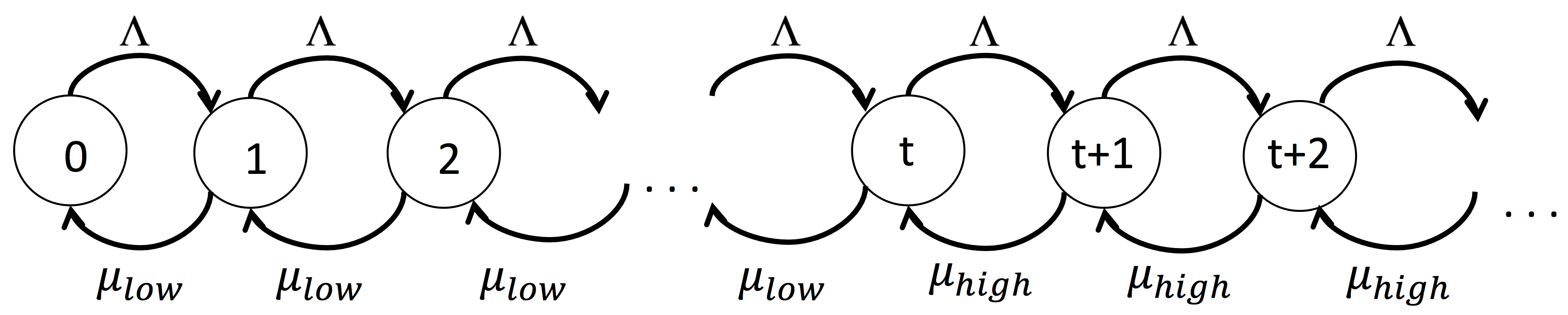}
\caption{A threshold chain with threshold state $t$ and arrival rate $\Lambda$.  For all states $i \leq t$, the service rate is $\mu_{low}$ and for all states $i > t$, the service rate is $\mu_{high}$.}
\label{thresholdChain}
\end{figure}

We can solve for the mean response time, $T$, in such a threshold chain, and find that 
\begin{align}\mathbb{E}[T]^{\mbox{thresh}}= \frac{1}{\Lambda}\cdot \biggl( t + \frac{\rho_{low}}{1-\rho_{low}} + \frac{1}{1-\rho_{high}}\label{eq-en-thresh}
+ \frac{1 +t - t\rho_{high}}{\rho_{high} -1+ \rho_{low}^t(\rho_{low}-\rho_{high})} \biggr).\end{align}

\subsubsection*{Constructing an Upper Bound}
We now construct a threshold chain which gives an upper bound (UB) on the mean response time under EQUI with arrival rate $\Lambda=cs(k^*) \mu$ (and thus $\rho_{k^*}=1$) and service rate $\mu$ per core.  To do this, we set $t=\lceil c(1+\epsilon) \rceil$, $\mu_{low}=s(n)\mu$ and $\mu_{high}=\lceil c(1+\epsilon) \rceil s\left(\frac{n}{\lceil c(1+\epsilon) \rceil}\right)\mu$ for any $\epsilon>0$.  Note that all departure rates in the UB chain are lower than those of EQUI (see Figure \ref{equiChain}), and thus this chain provides an upper bound on the mean response time under EQUI.  Our UB chain has:
$$\rho_{low} = \frac{\Lambda}{s(n) \mu} = \frac{c s(k^*) \mu}{s(n) \mu}$$
and 
$$\rho_{high} = \frac{\Lambda}{\lceil c(1+\epsilon) \rceil s\left(\frac{n}{\lceil c(1+\epsilon) \rceil}\right)\mu} =  \frac{c s(k^*) \mu}{\lceil c(1+\epsilon) \rceil s\left(\frac{n}{\lceil c(1+\epsilon) \rceil}\right)\mu}.$$
First note that by Lemma \ref{claim-increasing}, 
$$s(n) < cs(k^*) < \lceil c(1+\epsilon) \rceil s\left(\frac{n}{\lceil c(1+\epsilon) \rceil}\right).$$
Thus $\rho_{high} < 1 <\rho_{low}$.  We now apply \eqref{eq-en-thresh} to see that
\begin{align*}\lim_{n \rightarrow \infty} \mathbb{E}[T]^{\mbox{UB}} =  \lim_{n \rightarrow \infty} \frac{1}{\Lambda}\cdot\biggl(\lceil c(1+\epsilon) \rceil &+ \frac{\rho_{low}}{1-\rho_{low}} + \frac{1}{1-\rho_{high}}\\
 &+ \frac{1+\lceil c(1+\epsilon) \rceil(1-\rho_{high})}{\rho_{high} -1+ \rho_{low}^{\lceil c(1+\epsilon) \rceil}(\rho_{low}-\rho_{high})}\biggr).
\end{align*}
We can see that 
$$\lim_{n\rightarrow \infty} \frac{1}{\Lambda}\cdot \frac{\rho_{low}}{1-\rho_{low}}= \lim_{n\rightarrow \infty}   \frac{1}{\Lambda}\cdot\frac{1/\mu_{low}}{1/\Lambda - 1/\mu_{low}}= 0,$$
since $s(n)$ is bounded and $\mu_{low}$ therefore converges to a positive constant.
Furthermore, since $\rho_{high}$ converges to a constant less than 1, 
$$\lim_{n\rightarrow \infty}  \frac{1}{\Lambda}\cdot \frac{1}{1-\rho_{high}} = 0.$$
Finally, we see that
$$\lim_{n\rightarrow \infty} \frac{1+\lceil c(1+\epsilon) \rceil(1-\rho_{high})}{\Lambda\left(\rho_{high} -1+ \rho_{low}^{\lceil c(1+\epsilon) \rceil}(\rho_{low}-\rho_{high})\right)} =0,$$
since the numerator grows linearly in $n$ and the denominator grows exponentially in $n$.  Thus,
\begin{align*}
\lim_{n\rightarrow \infty} \mathbb{E}[T]^{\mbox{UB}}&= \lim_{n\rightarrow \infty} \frac{1}{\Lambda}\cdot\lceil c(1+\epsilon) \rceil\\
&=  \lim_{n\rightarrow \infty} \frac{1}{\Lambda}\cdot (c(1+\epsilon) + o(n))\\
&= (1+\epsilon)\frac{1}{s(k^*) \mu}.
\end{align*}

\subsubsection*{Constructing a Lower Bound}
Our argument for constructing a lower bound (LB) is largely the same.  We again assume $\Lambda=cs(k^*) \mu$  (and thus $\rho_{k^*}=1$) and we set $t=\lfloor c(1-\epsilon) \rfloor$, $\mu_{low}=\lfloor c(1-\epsilon) \rfloor s\left(\frac{n}{\lfloor c(1-\epsilon) \rfloor}\right)$, and $\mu_{high}=n\mu$.  Note that all departure rates in the LB chain are higher than those of EQUI, and thus this chain provides a lower bound on the mean response time under EQUI.  We now have
$$\rho_{low} = \frac{\Lambda}{\lfloor c(1-\epsilon) \rfloor s\left(\frac{n}{\lfloor c(1-\epsilon) \rfloor}\right)\mu} =  \frac{c s(k^*) \mu}{\lfloor c(1-\epsilon) \rfloor s\left(\frac{n}{\lfloor c(1-\epsilon) \rfloor}\right)\mu}$$
and
$$\rho_{high} = \frac{\Lambda}{n \mu} = \frac{c s(k^*) \mu}{n \mu}=\frac{s(k^*) \mu}{k^* \mu}.$$
By Lemma \ref{claim-increasing} we again see that $\rho_{high} < 1 < \rho_{low}$.  We apply \eqref{eq-en-thresh} to see that
\begin{align*}\lim_{n \rightarrow \infty} \mathbb{E}[T]^{\mbox{LB}} = \frac{1}{\Lambda}\cdot\biggl(\lim_{n \rightarrow \infty} \lfloor c(1-\epsilon) \rfloor  &+ \frac{\rho_{low}}{1-\rho_{low}} + \frac{1}{1-\rho_{high}}\\
 &+ \frac{1+\lfloor c(1-\epsilon) \rfloor(1+\rho_{high})}{\rho_{high} -1+ \rho_{low}^{\lfloor c(1-\epsilon) \rfloor }(\rho_{low}-\rho_{high})}\biggr).
\end{align*}
Since $\rho_{low}$ converges to a constant greater than 1 and $\rho_{high}$ is a constant,
$$\lim_{n\rightarrow \infty} \frac{1}{\Lambda}\cdot\frac{\rho_{low}}{1-\rho_{low}} = 0 \mbox{ and}\lim_{n\rightarrow \infty} \frac{1}{\Lambda}\cdot\frac{1}{1-\rho_{high}} = 0.$$
Finally, we see that
$$\lim_{n\rightarrow \infty} \frac{1+\lfloor c(1-\epsilon) \rfloor(1+\rho_{high})}{\Lambda\left(\rho_{high} -1+ \rho_{low}^{\lfloor c(1-\epsilon) \rfloor }(\rho_{low}-\rho_{high})\right)} = 0,$$
since the numerator grows linearly in $n$ and the denominator grows exponentially in $n$.  Thus,
\begin{align*}
\lim_{n\rightarrow \infty} \mathbb{E}[T]^{\mbox{UB}}&= \lim_{n\rightarrow \infty} \frac{1}{\Lambda}\cdot\lfloor c(1-\epsilon) \rfloor\\
&=  \lim_{n\rightarrow \infty} \frac{1}{\Lambda}\cdot (c(1-\epsilon) - o(n))\\
&= (1-\epsilon)\frac{1}{s(k^*) \mu}.
\end{align*}

We have therefore shown that
$$(1-\epsilon)\frac{1}{s(k^*) \mu} \leq \lim_{n\rightarrow \infty} \mathbb{E}[T]^{\mbox{EQUI}}_{\rho_{k^*}=1} \leq (1+\epsilon)\frac{1}{s(k^*) \mu}$$
for any $0 < \epsilon \leq 1$ and we have thus shown \eqref{eq:jsq:opt} as desired.
\end{proof}

\section{Multiple Speedup Functions}\label{sec-multiple}

\begin{figure*}[h!t]
\centering
\subfigure[EQUI vs. OPT]{
\includegraphics[width=.45\textwidth]{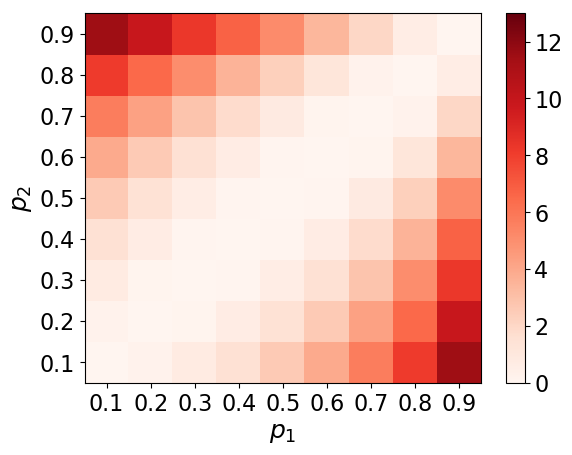}
\label{fig:heat:equi}
}
\subfigure[\gstar{} vs. OPT]{
\includegraphics[width=.45\textwidth]{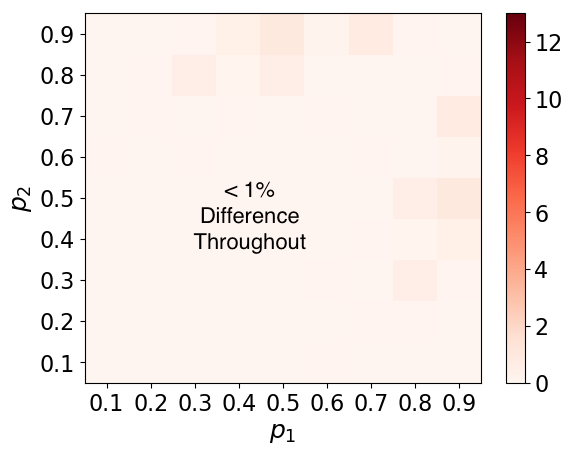}
\label{fig:heat:greedy}
}
\vspace{-0.25in}
\caption{Heat maps showing the percentage difference in the mean response time in the system, $\mathbb{E}[T]$, between (a) EQUI and OPT and between (b) \gstar{} and OPT, in the case of two speedup functions where $s_1$ and $s_2$ are Amdahl's law with parameters $p_1$ and $p_2$ respectively.  Here $\mathbb{E}[X_1]=\mathbb{E}[X_2]=\frac{1}{2}$ and $\Lambda_1=\Lambda_2=5$.  The axes represent different values of $p$ for each class.  \gstar{}, EQUI, and OPT were evaluated numerically using the MDP formulation given in Section \ref{mdp-opt}.  These heat maps look similar under various values of $\Lambda_1$ and $\Lambda_2$.} 
\label{fig:heat}
\vspace{-0.15in}
\end{figure*}

Thus far, we have assumed that jobs are homogeneous with respect to speedup.  In this section, we consider the case where jobs may have different speedup functions.  To facilitate dealing with multiple speedup functions, we will assume throughout this section that job sizes are exponentially distributed (see Section \ref{sec:opt:general} for a discussion of why general job size distributions are outside the scope of this paper when dealing with optimality).

We will see in Section \ref{equi-not-opt} that, in the case of multiple speedup functions, EQUI is no longer the optimal policy.  In Section \ref{greedy}, we propose a class of policies called GREEDY which maximize the departure rate in every state.  We then describe the optimal GREEDY policy, \gstar{} in Section \ref{gstar}.  While \gstar{} is not optimal in general (see Section \ref{greedy-near-opt}), we show that \gstar{} performs near-optimally in a wide range of settings (Sections \ref{mdp-opt} and \ref{greedy-near-opt}).  Finally in Section \ref{fixed-width-multi}, we return to fixed-width policies and explain why they are insufficient when there are multiple speedup functions.

\subsection{Why Multiple Speedup Functions}
There are situations in which it would be reasonable to expect all jobs to follow a single speedup function, such as when all jobs are instances of a single application.  The PARSEC-3 benchmark provides many examples of workloads for which this is the case \cite{zhan2017parsec3}.  In practice, however, it may be the case that there are 2 or more classes of jobs, each with a unique speedup function reflecting the amount of sequential work, number of IO operations, and communication overhead that its jobs will experience when run across multiple cores.

We consider the case where jobs may belong to one of two \emph{classes}, each of which has its own speedup function, $s_i$ (the classes may also have arrival rates, $\Lambda_i$, but we assume all job sizes to be exponentially distributed with rate $\mu= 1/\mathbb{E}[X]$).  Without loss of generality, we assume that class 1 jobs are less parallelizable than class 2 jobs: $s_1(k) < s_2(k)$ for $k > 1$.  For example, class 1 jobs could follow Amdahl's law with $p=.5$ while class 2 jobs follow Amdahl's law with $p=.75$ (see Figure \ref{figAmdahl}).  As usual, our goal is to describe and analyze scheduling policies which minimize overall mean response time across all jobs.  We will assume that a scheduling policy can differentiate between job classes when making scheduling decisions.
\vspace{-.1in}
\subsection{EQUI is No Longer Optimal}\label{equi-not-opt}
We have already seen that EQUI is optimal when jobs are homogeneous with respect to speedup. One might assume that, since EQUI bases its decisions on the number of jobs in the system rather than the jobs' speedup functions, EQUI could continue to perform well when there are multiple speedup functions.  However, it turns out that EQUI's performance is suboptimal even when there are just two speedup functions (see Figure \ref{fig:heat:equi}).  While EQUI's performance is actually close to optimal in the cases where $s_1$ and $s_2$ are similar, we see that EQUI's performance relative to the optimal policy becomes worse as the difference between the speedup functions increases.

To see why EQUI is suboptimal in this case, recall that EQUI's optimality stems from the fact that it maximizes the rate of departures in every state when jobs follow a single speedup function (see proof of Theorem \ref{equi-thm-global-opt}).  When jobs are permitted to have different speedup functions, maximizing the rate of departures will require allocating more cores to class 2 jobs and fewer cores to class 1 jobs.
%\vspace{-.1in}
\subsection{A GREEDY Class of Policies}\label{greedy}
We have seen that EQUI fails to maximize the rate of departures when there are multiple speedup functions.  Would a policy that maximizes the total rate of departures be optimal in this case?  We define the GREEDY class of policies to be the policies which achieve the maximal total rate of departures in every state. 

To describe the policies in GREEDY, we again consider the case where jobs belong to one of two job classes.  For any state $(x_1,x_2)$ where there are $x_1$ class 1 jobs and $x_2$ class 2 jobs, attaining the maximal rate of departures can be thought of as a two step process.  \emph{First}, a policy must decide how many cores, $a_1$, to allocate to the $x_1$ class 1 jobs.  The remaining $a_2=n-a_1$ cores will be allocated to class 2 jobs.  \emph{Second}, the policy must decide how to divide the $a_1$ cores among the class 1 jobs and the $a_2$ cores among the class 2 jobs.  We have seen that, when dividing cores among a set of jobs with a single speedup function, EQUI maximizes the total rate of departures of this set of jobs.  For a given choice of $a_1$, the $a_1$ cores should thus be evenly divided among the class 1 jobs and the $a_2$ cores should be evenly divided among the class 2 jobs in order to maximize the total rate of departures.  Thus, only the first decision remains.

To find an allocation of cores which maximizes the rate of departures, we first define $\beta(x_1,x_2)$, the maximum rate of departures from the state $(x_1, x_2)$:
\begin{equation}\label{eq:beta}
\beta(x_1, x_2) = \max_{\alpha \in [0,n]} x_1s_1\left(\frac{\alpha}{x_1}\right)\mu+x_2s_2\left(\frac{n-\alpha}{x_2}\right)\mu.
\end{equation}
We must then choose $a_1$ such that
\begin{equation}\label{eq:greedy}
a_1 \in \left\{\alpha : x_1s_1\left(\frac{\alpha}{x_1}\right)\mu + x_2 s_2\left(\frac{n-\alpha}{x_2}\right)\mu=\beta(x_1, x_2)\right\}.
\end{equation}
This same two-step process will generalize to the case when jobs follow more than 2 speedup functions.

Crucially, note that GREEDY is truly a \emph{class} of policies, since, in a given state, there may be multiple choices of $a_1$ which satisfy \eqref{eq:greedy}.  That is, there could be multiple allocations which achieve the maximal total rate of departures.  For example, consider a system with 4 cores where both jobs classes have the same service rate $\mu=\frac{1}{\mathbb{E}[X]}$.  If there are 4 class 1 jobs and 4 class 2 jobs, any choice of $a_1\in[0,4]$ results in the maximal rate of departures, $n\mu$.  This begs the question of \emph{which} policy from the GREEDY class achieves the best performance.

\subsection{The Best GREEDY Policy: \gstar{}}\label{gstar}
We now define \gstar{}, a policy which dominates all other GREEDY policies with respect to mean response time.  Consider two GREEDY policies, $P_1$ and $P_2$.  In any state $(x_1, x_2)$, both policies achieve the same maximal rate of departures.  However, $P_1$ might achieve this rate by having a higher departure rate for class 1 jobs and a lower departure rate for class 2 jobs as compared to $P_2$.  If class 1 jobs are less parallelizable than class 2 jobs, we say that $P_1$ ``defers parallelizable work'' in this state, which is a strategy that could benefit $P_1$ in the future.

The \gstar{} policy is the GREEDY policy which in all states opts to defer parallelizable work when possible.
Specifically, in the case of two job classes:  \gstar{} allocates $a_1^*$ cores to class 1 jobs (the less parallelizable class), where $a_1^*$ is the maximum value of $a_1$ satisfying \eqref{eq:greedy}.  That is,
$$a_1^*= \max{\left\{\alpha:x_1s_1\left(\frac{\alpha}{x_1}\right)\mu+x_2s_2\left(\frac{n-\alpha}{x_2}\right)\mu=\beta(x_1, x_2)\right\}}.$$

In other words, $a_1^*$ allows \gstar{} to attain the maximal rate of departures while \emph{also} maximizing the rate at which class 1 jobs are completed.  Theorem \ref{thm-gstar} shows that \gstar{} dominates all other GREEDY policies.
\begin{theorem}\label{thm-gstar}
For any GREEDY policy, $P$, $$\mathbb{E}[T]^{\gstars} \leq \mathbb{E}[T]^P.$$
\end{theorem}
\begin{proof}
Consider the performance of \gstar{} and $P$ on the state space $\mathcal{S} = \{(x_1, x_2): x_1, x_2 \in \mathbb{N}\}$.  We will use the technique of precedence relations (see, e.g., \cite{adan1994upper,buvsic2012comparing}) to compare the mean number of customers, $\mathbb{E}[N]$, under \gstar{} to that under $P$.  This requires that we define a value function for $P$, $V^P(x_1, x_2)$, and a cost function, $c(x_1, x_2)$.  We define the cost of being in state $(x_1, x_2)$ to be
$$c(x_1, x_2) =x_1 +x_2$$
so that the average cost of performing policy $P$ is equal to the mean number of customers, $\mathbb{E}[N]$.  We then define $V^P(x_1, x_2)$ to be the asymptotic total difference in the accrued cost under $P$ when starting in state $(x_1, x_2)$ as opposed to some designated reference state (see Appendix \ref{pf:lemma:v} for details).  We require the following lemma which establishes a useful property of $V^P$.

\begin{lemma}\label{lemma:v}
For any GREEDY policy, $P$, and any $(x_1, x_2) \in \mathcal{S}$, 
$$V^P(x_1+1, x_2) > V^P(x_1, x_2+1).$$
\end{lemma} 
\begin{proof}
See Appendix \ref{pf:lemma:v} for the proof of Lemma \ref{lemma:v}.
\end{proof}
We now prove Theorem \ref{thm-gstar} by contradiction.  We begin by assuming that there exists a GREEDY policy $P \neq \gstars$ which is optimal in terms of $\mathbb{E}[N]$ and thus $\mathbb{E}[N]^P \leq \mathbb{E}[N]^{P'}$ for any GREEDY policy $P'$.  Since $P \neq \gstars$, there exists some state, $(x_1, x_2) \in \mathcal{S}$ where $P$ and \gstar{} take different actions.  Note that both $x_1$ and $x_2$ must be non-zero in this state, because otherwise there is only one action which will achieve the maximal rate of departures, and $P$ and \gstar{} must therefore take the same action.  

We now consider a policy $P'$ which takes the same action as \gstar{} in state $(x_1, x_2)$, and the same action as $P$ in every other state.  We can apply the technique of precedence relations described in \cite{adan1994upper,buvsic2012comparing} to show that $\mathbb{E}[N]^{P'} < \mathbb{E}[N]^P$.  We start by defining $\gamma^{Q}_i(x_1,x_2)$ to be the rate of departures of class $i$ jobs from the state $(x_1, x_2)$ under policy $Q$.  $P'$ can now be obtained from $P$ by taking $\gamma^{P}_2(x_1, x_2) - \gamma^{\gs}_2(x_1, x_2)$ away from the total completion rate of class 2 jobs and adding $\gamma^{\gs}_1(x_1, x_2) - \gamma^{P}_1(x_1, x_2)$ to the total completion rate of class 1 jobs, where $G^*$ denotes \gstar{}.  Theorem 3.1 in \cite{buvsic2012comparing} tells us that $\mathbb{E}[N]^{P'} < \mathbb{E}[N]^P$ if: 
\begin{align*}
\left(\gamma_1^{\gs}(x_1, x_2)-\gamma_1^{P}(x_1, x_2)\right)V^P(x_1-1, x_2)  < \left(\gamma_2^P(x_1, x_2)-\gamma_2^{\gs}(x_1, x_2)\right)V^P(x_1, x_2-1).
\end{align*}

To see that this property holds, first note that 
\begin{align*}
&\left(\gamma_1^{\gs}(x_1, x_2)-\gamma_1^{P}(x_1', x_2')\right) - \left(\gamma_2^P(x_1, x_2)-\gamma_2^{\gs}(x_1, x_2)\right)\\
 =& \left(\gamma_1^{\gs}(x_1, x_2) + \gamma_2^{\gs}(x_1, x_2)\right) - \left(\gamma_1^{P}(x_1, x_2) +\gamma_2^P(x_1, x_2)\right)=0
\end{align*} 
since \gstar{} and $P$ have the same (maximal) total rate of departures in every state.  Thus,
$$\gamma_1^{\gs}(x_1, x_2)-\gamma_1^{P}(x_1, x_2) = \gamma_2^P(x_1, x_2)-\gamma_2^{\gs}(x_1, x_2).$$
By Lemma \ref{lemma:v}, we know that
$$V^P(x_1-1, x_2) < V^P(x_1, x_2-1).$$
This implies that $\mathbb{E}[N]^{P'} < \mathbb{E}[N]^P$ which contradicts our assumption that $P$ is optimal in terms of $\mathbb{E}[N]$.  By Little's Law, we can reformulate this in terms of $\mathbb{E}[T]$.
\end{proof}

While \gstar{} is the best GREEDY policy, it will turn out that it is not optimal (see Section \ref{greedy-near-opt}).  Hence, we now turn our attention to computing the optimal policy.

\subsection{Computing the Optimal Policy}\label{mdp-opt}
The optimal policy, OPT, must not only consider the current state of the system when choosing how to determine the best partition $(a_1, a_2)$, but must also consider the probabilities of transitioning to future states as well.  To find a policy which balances this tradeoff between performance in the current state and future states, we formulate the problem as a Markov Decision Process (MDP).

We will consider an MDP with state space $\mathcal{S} = \{(x_1, x_2): x_1, x_2 \in \mathbb{N}\}$, where $x_i$ represents the number of class $i$ jobs in the system.  The action space in any state is given by $\mathcal{A} = \left\{(a_1, a_2): a_1+a_2 = n\right\}$.  Let the arrival rate of class $i$ jobs be given by $\Lambda_i$.  Recall that, given an allocation of $a_i$ cores to $x_i$ type $i$ jobs, it is optimal to run the $x_i$ jobs on these cores using EQUI.  Thus, given a state $(x_1, x_2)$ and an action $(a_1, a_2)$, the total departure rate of class $i$ jobs from the system is given by 
$$\mu_i(a_i, x_i) := \min\{a_i, x_i\}\mu s\left(\max\left\{1,\frac{a_i}{x_i}\right\}\right).$$
 We choose the cost function 
 $$c(x_1, x_2) = x_1+x_2$$
 such that the \emph{average cost per period} equals the average number of jobs in the system, $\mathbb{E}[N]$.  We uniformize the system at rate 1 (always achievable by scaling time) and find that Bellman's optimality equations \cite{puterman2014markov}  for this MDP are given by
\begin{equation*}
\E{N^{OPT}}+V^{OPT}(x_1,x_2) = A^{OPT}(x_1, x_2) + H^{OPT}(x_1, x_2), 
\end{equation*}

where 
\begin{align}
A^{OPT}(x_1, x_2) =\ c(x_1,x_2) &+ \Lambda_1\left(V^{OPT}(x_1+1,x_2)-V^{OPT}(x_1, x_2)\right) \notag\\
&+ \Lambda_2 \left(V^{OPT}(x_1, x_2+1)-V^{OPT}(x_1, x_2)\right), \label{eq:A}
\end{align}
\begin{align*}
H^{OPT}(x_1, x_2) = V^{OPT}(x_1, x_2) + \min_{(a_1, a_2) \in \mathcal{A}}\biggl\{ \hspace{1.4in}
\end{align*}
\begin{align}
 \mu_1(a_1, x_1)&\left(V^{OPT}\left((x_1-1)^+, x_2\right)-V^{OPT}(x_1, x_2)\right) \notag\\
+ \mu_2(a_2, x_2)&\left(V^{OPT}\left(x_1, (x_2-1)^+\right)-V^{OPT}(x_1, x_2)\right)\biggr\}. \label{eq:H}
\end{align}
Here, the value function $V^{OPT}(x_1,x_2)$ denotes the asymptotic total difference in accrued costs when using the optimal policy and starting the system in state $(x_1,x_2)$ instead of some reference state.  While these equations are hard to solve analytically, the optimal actions can be obtained numerically by defining $$V^{OPT}_{n+1}(x_1, x_2) = A^{OPT}_n(x_1, x_2)+H^{OPT}_n(x_1, x_2)$$ 
with $V^{OPT}_0(\cdot, \cdot) = 0$.  
Here, $A_n^{OPT}$ and $H_n^{OPT}$ are defined as in \eqref{eq:A} and \eqref{eq:H}, but in terms of $V_n^{OPT}$.  We can then perform \emph{value iteration} to extract the optimal policy \cite{lippman1973semi}.  We use the results of this value iteration to compare the performance of OPT to several other policies in Figure \ref{fig:heat} and Figure \ref{fig:heat:static}. 

%Note that \eqref{eq:H} helps explain why using the maximal rate of departures is optimal for EQUI when there was only one speedup function.  In this case, there would be only two terms -- the term for completing a job and the term for remaining in the current state.  Hence, it is optimal to simply use all $n$ cores and maximize probability of completing a job.  By adding a second class of jobs, we add a third term to this equation.  We have gone from a one-dimensional state space into a two-dimensional state space and our future performance will depend on \emph{which} job we complete next.
Note that using the same MDP formulation when there exists only one speedup function results in a much simpler expression for $V^{OPT}$ which clearly yields EQUI.

\subsection{\gstar{} is Near-Optimal}\label{greedy-near-opt}
Surprisingly, even \gstar{} is not optimal for minimizing mean response time as shown in Figure \ref{fig:heat:greedy} (although it is always within 1\% of OPT in the figure).  The same intuition that led us to believe that \gstar{} is the best GREEDY policy can be used to explain why \gstar{} is not optimal.  We have already seen that deferring parallelizable work is advantageous to \gstar{}.  However, we were only comparing \gstar{} to policies with maximal overall departure rate.  It turns out that the advantage of deferring parallelizable work can be so great that a policy stands to benefit from achieving a submaximal overall rate of departures in order to defer parallelizable work.  

%In Figure \ref{fig:heat:greedy} we can see that \gstar{} performs well compared to OPT over a wide range of speedup functions.  While this figure is only for one set of arrival rates, results for other arrival rates looked very similar.  Intuitively, \gstar{} performs well because performance in the current state is weighted much more heavily than future performance.  For example, the cost of being in a state with 2 jobs in the system is twice as high as the cost of having 1 job in the system, and thus it often the right decision to simply leave the current state as quickly as possible.

\subsection{Does Fixed-Width Scheduling Work?}\label{fixed-width-multi}

\begin{figure}[t!]
\centering
\includegraphics[width=.45\textwidth]{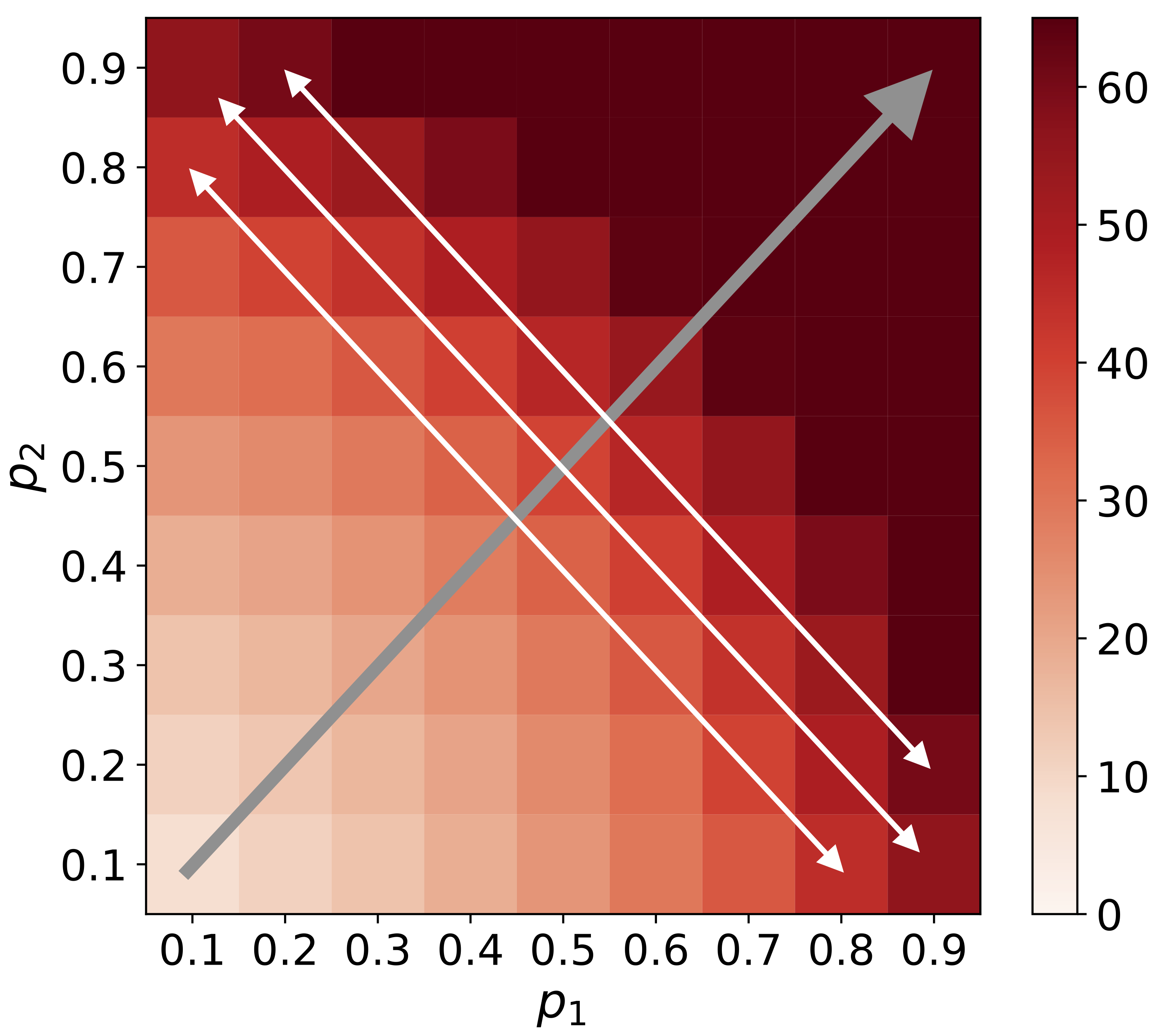}
\caption{Heat map showing the percentage difference in the mean response time, $\mathbb{E}[T]$, between JSQ-Chunk with the optimal $k^*$ and OPT , in the case of two speedup functions where $s_1$ and $s_2$ are Amdahl's law with parameters $p_1$ and $p_2$ respectively.  Here $\mathbb{E}[X_1]=\mathbb{E}[X_2]=1$ and $\Lambda_1=\Lambda_2=2$.  The axes represent the value of $p_i$ for each class.  OPT was evaluated numerically using the MDP formulation given in Section \ref{mdp-opt} and the results for JSQ-Chunk come from analysis.} 
\label{fig:heat:static}
\end{figure}

We saw that, when jobs follow a single speedup function, the JSQ-Chunk policy with the optimal chunk size $k^*$ will often achieve near-optimal performance (see Section \ref{jsq:near:opt}).  Is this still true when jobs are permitted to have different speedup curves?  Figure \ref{fig:heat:static} compares JSQ-Chunk with OPT for the case of two speedup functions.  We see two trends:
\vspace{.1in}

\noindent\textbf{Trend 1: Along the thick arrow in Figure \ref{fig:heat:static}, JSQ-Chunk becomes further from OPT as $p_1 + p_2$ increases.}  

Trend 1 makes intuitive sense since, when jobs are more parallelizable, OPT will be able to more effectively exploit this parallelism while JSQ-Chunk will be limited by being restricted to use a single $k^*$ for each job.  Nonetheless, trend 1 becomes irrelevant as the number of cores, $n$, increases, since JSQ-Chunk converges to OPT when $p_1$ equals $p_2$ (single speedup function).  Thus, we are more interested in trend 2.

\vspace{.1in}
\noindent\textbf{Trend 2:  Along any thin diagonal in Figure \ref{fig:heat:static} ($p_1 + p_2 = \mbox{constant}$), JSQ-Chunk becomes further from OPT as we move outward on the diagonal.}

Trend 2 follows from two observations.   First, Theorem \ref{thm:jsq:const} proves that JSQ-Chunk's performance is fixed along these diagonals.  Second, we have reason to believe that the mean response time under OPT should decrease as we move further outward along these diagonals (see Observation \ref{obs:opt}).  Together, these observations explain trend 2.  The rest of this section discusses Theorem \ref{thm:jsq:const} and Observation \ref{obs:opt}.

\begin{theorem}\label{thm:jsq:const}
Given two speedup functions, $s_1$ and $s_2$, where $s_1$ follows Amdahl's law with parameter $p_1$ and $s_2$ follows Amdahl's law with parameter $p_2$, and $\Lambda_1 = \Lambda_2$, the mean response time under JSQ-Chunk with the optimal $k^*$ is constant for any $(p_1, p_2)$ such that $p_1 + p_2=c$, where $c$ is a constant. 
\end{theorem}
\begin{proof}
In the case of two speedup functions, under JSQ-Chunk with level of parallelization, $k$, $X_k$ from \eqref{eq:service} becomes:
\vspace{-0.1in}
\begin{equation*}X_k = \begin{cases}
\frac{X}{s_1(k)} & \mbox{w.p.} \quad \frac{\Lambda_1}{\Lambda_1 + \Lambda_2}\\
\frac{X}{s_2(k)} & \mbox{w.p.} \quad \frac{\Lambda_2}{\Lambda_1 + \Lambda_2} \end{cases}.
\end{equation*}

We now prove that, for any given $k$, $\mathbb{E}[X_k]$ is constant whenever $p_1 + p_2 =c$.
By definition, when $\Lambda_1 = \Lambda_2$,
\begin{equation}\mathbb{E}[X_k]=\frac{1}{2} \cdot \frac{\mathbb{E}[X]}{s_1(k)} +\frac{1}{2} \cdot \frac{\mathbb{E}[X]}{s_2(k)}.\label{eq:exk}\end{equation}
Since $s_1$ and $s_2$ are both instances of Amdahl's law, we can use a property of Amdahl's law that states 
\begin{equation}\frac{1}{2} \cdot \frac{1}{s_1(k)} +\frac{1}{2} \cdot \frac{1}{s_2(k)} = \frac{1}{s_3(k)}, \label{eq:amdahl}\end{equation}
where $s_3(k)$ is the speedup function for Amdahl's law with parameter $p_3 = \frac{p_1 + p_2}{2}$.  Combining \eqref{eq:exk} and \eqref{eq:amdahl} we have 
$$\mathbb{E}[X_k]=\frac{\mathbb{E}[X]}{s_3(k)},$$
which is a constant, provided that $p_1 + p_2 =c$.  

As stated in Section \ref{jsq}, the performance of JSQ-Chunk depends \emph{only} on $\mathbb{E}[X_k]$.  Hence, along any diagonal where $p_1+p_2 = c$, the mean response time under JSQ-Chunk with chunk size $k$ remains constant.  Since any choice of $k$ leads to constant performance along the diagonal, setting $k$ to $k^*$ will also keep JSQ-Chunk constant along this diagonal.  Hence, the mean response time under JSQ-Chunk with the optimal $k^*$ is constant when $p_1 + p_2 =c$.
\end{proof}

\begin{observation}\label{obs:opt}
Along any diagonal where $p_1 + p_2 =c$, where $c$ is a constant and $\Lambda_1 = \Lambda_2$, we expect the mean response time under OPT to decrease as $\mid p_1 - p_2 \mid$ increases.  
\end{observation}
Although the mean job size is constant along this diagonal, when $\mid p_1 - p_2 \mid$ is higher, the optimal policy has additional information about which jobs will take longer to run.  In effect, this allows OPT to favor jobs which benefit more from parallelization and leads to a lower overall mean response time.  A similar effect was observed in \cite{DBLP:journals/ior/XuSS15}, where assigning jobs different service rates lowered optimal mean response time.

\begin{remark}
Although Theorem \ref{thm:jsq:const} and Observation \ref{obs:opt} assume that $\Lambda_1 = \Lambda_2$, they are easily generalized to cases where the arrival rates are not equal.  In these cases, the diagonals along which $\mathbb{E}[X_k]$ will be constant will have a different slope.  Nonetheless, trend 2 will still occur along these diagonals.
\end{remark}

Unlike trend 1, which disappears as the number of cores, $n$, increases, we do not expect trend 2 to vanish.  This is supported by our evaluation of OPT under higher values of $n$ (not shown).  We thus conclude that JSQ-Chunk does not perform near-optimally when jobs follow multiple speedup functions.
%\vspace{-0.1in}
\section{Conclusion and Future Work}

\noindent\textbf{Summary:} This paper introduces the question of how to allocate cores to jobs in a stochastic model of a multi-core machine where jobs have sublinear speedup functions.  While it is typical for the \emph{user} to specify the desired level of parallelization for her jobs, this paper instead proposes that allowing the \emph{system} to schedule jobs can benefit overall mean response time.  In the case where all jobs are malleable and follow the same speedup function, we prove that the well-known EQUI policy is optimal when job sizes are exponentially distributed.  In the more practical setting where jobs are moldable rather than malleable, we prove that, surprisingly, one can still achieve near-optimal performance by using the optimal \emph{fixed} level of parallelization, $k^*$.  We show how to analytically determine $k^*$ as a function of system load, the speedup curve, the job size distribution, the number of cores, and the dispatching policy.

When jobs are permitted to have multiple speedup functions, the question of optimal scheduling becomes even harder.  We show that EQUI is no longer optimal for scheduling malleable jobs when jobs follow multiple speedup functions.  We find that the optimal policy (OPT) must balance the tradeoff between maximizing the total rate of departures in every state and deferring parallelizable work to minimize wasting cores in the future.  We provide an MDP based formulation of OPT.  As an alternative, we introduce a very simple policy, \gstar{}, which can be easily implemented and performs near-optimally in practice.  In the more practical setting where jobs are moldable but not malleable, we again consider fixed-width scheduling but find that it is too inflexible to balance the tradeoff considered by the optimal policy.\vspace{.2in}

\noindent\textbf{Limitations and Future Work:} We see this paper as opening up a new line of research within the field of parallel scheduling.  One limitation of our work is that handling multiple speedup functions leads to a complex optimization problem.  It may be easier to consider a continuous regime where jobs draw their speedup functions from a distribution.  A second limitation of our model is that we assume a stream of jobs where jobs have no internal dependency structure.  In reality, as the SPAA community has pointed out, a single job usually consists of multiple interdependent tasks.  The problem of how to schedule a stream of jobs which have these dependency structures is open, even in the case of a single processor \cite{ziv}.  Hence, scheduling these jobs on multiple cores presents an exciting direction for future work.  Finally, while we are the first to provide optimality proofs for policies like EQUI, our optimality results rely on job sizes being exponentially distributed.  As explained, moving to general job size distributions makes scheduling much more complex by forcing policies to consider the distribution of residual job lifetimes.   Although this paper does not model all of these aspects, we hope to provide a starting point for further research in this area.

\clearpage
\section{Appendix}
\subsection{Proof of Lemma \ref{claim-increasing}}\label{pf-increasing}
LEMMA 2.2 \emph{For any concave, sublinear function, $s$, the function $i \cdot s\left(\frac{n}{i}\right)$ is increasing in $i$ for all $i<n$, and is non-decreasing in $i$ for all $i\geq n$.}
\begin{proof}
To see that $i \cdot s\left(\frac{n}{i}\right)$ is increasing in $i$ when $i<n$, we can consider the following difference for any $\delta > 0$:
$$i\cdot(1+\delta)s\left(\frac{n}{i\cdot(1+\delta)}\right) - i \cdot s\left(\frac{n}{i}\right)= i\left((1+\delta)s\left(\frac{n}{i\cdot(1+\delta)}\right) - s\left(\frac{n}{i}\right)\right).$$
Since $\frac{n}{i}>1$, $s$ increases sublinearly, and $(1+\delta)s\left(\frac{n}{i\cdot(1+\delta)}\right) > s\left(\frac{n}{i}\right)$.  Thus
$$i\cdot(1+\delta)s\left(\frac{n}{i\cdot(1+\delta)}\right) - i \cdot s\left(\frac{n}{i}\right) > 0$$
and $i \cdot s\left(\frac{n}{i}\right)$ is increasing in $i$.

For any $i\geq n$, we have assumed $s\left(\frac{n}{i}\right)=\frac{n}{i}$.  Thus, 
$$i \cdot s\left(\frac{n}{i}\right)=n \indent \forall i \geq n,$$
which is non-decreasing in $i$.
\end{proof}

\subsection{Mixed-Random-Chunk}\label{sec-mrc}
In this section, we explore Mixed-Random-Chunk (MRC), which uses two chunk sizes: $k_1$ and $k_2$.  We will also restrict ourselves to consider policies which balance load between chunks according to the size of a chunk.  That is, given that jobs are dispatched to an average of $k$ cores, the average per-core arrival rate should be $\Lambda k/n$.

Under Mixed-Random-Chunk, we group $a_1$ of the $n$ cores into chunks of size $k_1$. We call these chunks size-$k_1$ chunks. The remaining $a_2 := n-a_1$ cores are grouped into size-$k_2$ chunks.  Like before, we assume that $k_1$ divides $a_1$, and $k_2$ divides $a_2$.

 \begin{figure}[h!]
\centering
\includegraphics[width=.3\textwidth]{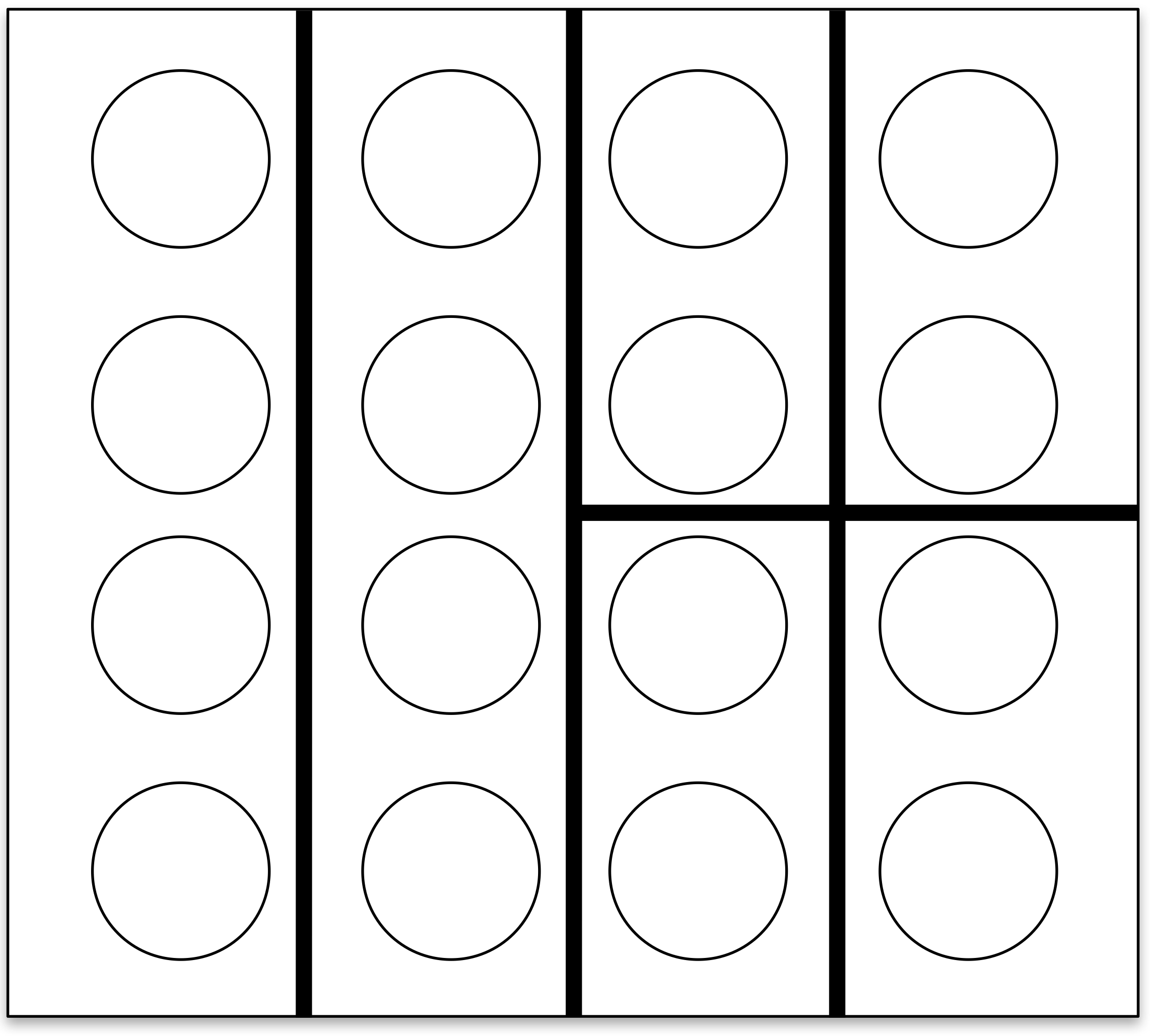}
\caption{A system with $n=16$ under a Mixed-Random-Chunk policy.  Here, $a_1=8$ of the cores are grouped into chunks of size $k_1=4$, and $a_2=8$ of the cores are grouped into chunks of size $k_2=2$.}
\label{fig:mixed-chunks}
\end{figure}

When a job arrives, we choose a core uniformly at random (with probability $\frac{1}{n}$).   The job is then parallelized across all cores in the chosen core's chunk.  We can see via Poisson splitting that the arrivals to a size-$k_i$ chunk form a Poisson process with rate $\frac{\Lambda k_i}{n}$.  Under this policy, the level of parallelization is defined to be $k=(a_1k_1+a_2k_2)/n$.  Hence, the average per-core arrival rate is
$$\Lambda_k = \frac{a_1}{n} \cdot \frac{\Lambda k_1}{n} + \frac{a_2}{n} \cdot \frac{\Lambda k_2}{n}= \frac{\Lambda k}{n},$$
and Mixed-Random-Chunk obeys the load balancing property outlined above.

\subsubsection*{Mean Response Time}
To derive the mean response time $\E{T}$ under Mixed-Random-Chunk, we must condition on the type of chunk to which an arrival is sent.  We denote by $E_i$ the event that the arrival is sent to a size $k_i$ chunk.  Recall that the arrival rate to a core in a size-$k_i$ chunk is $\Lambda_{k_i} = \Lambda k_i /n$ and the size of a job piece at this core is distributed as $X_{k_i} = X/\text{s(}k_i)$.  This information suffices to calculate the mean response time of a job sent to a size-$k_i$ chunk, $\E{T \mid E_i}$ (see Theorem \ref{rc-thm-mrt} and its proof). Next, note that a job will be sent to a size-$k_i$ chunk with probability $a_i/n$. Hence, we can find the mean response time for the entire system by conditioning:
\begin{equation}\label{eq:ET}
    \E{T}^{\mbox{\scriptsize MRC}} =  \frac{a_1}n\E{T \mid E_1} +\frac{a_2}n \E{T \mid E_2}.
\end{equation}

\subsubsection*{Optimizing $a_1$ and $a_2$}\label{sec:optET}
In \eqref{eq:ET} observe that $a_1$ and $a_2$ are decision variables that can be chosen to minimize mean response time.  Surprisingly, we find in Theorem \ref{sec:optETThm} that the optimal choice is to have only \emph{one} chunk size.

\begin{theorem}\label{sec:optETThm}
The mean response time under Mixed-Random-Chunk is minimized when the number of size-$k_1$ chunks, $a_1$, is either $0$ or $n$.
\end{theorem}
\begin{proof}
Since $a_2 = n-a_1$, \eqref{eq:ET} can be rewritten as
\begin{equation*}
\E{T} = \frac{a_1}n\left(\E{T \mid E_1} - \E{T \mid E_2}   \right)+\E{T \mid E_2}.
\end{equation*}
This expression is of the form $\E{T} = a_1x+y$, and thus the mean response time is linear in $a_1$.  We can also see that $x$ and $y$ are independent of $a_1$ and $a_2$.  This follows since, for any core belonging to a chunk of size $k_i$, neither the service requirement, $X_{k_i}=X_{k_i}/s(k_i)$, nor the arrival rate $\Lambda_{k_i}=\Lambda k_i/n$ depend on $a_i$.  Thus, our linear expression is either increasing or decreasing in $a_1$, and will thus be minimized at a boundary where $a_1=0$ or $a_1=n$.
\end{proof}

\subsubsection*{Variance of Response Time}\label{sec:VarT}
Not only is mean response time not improved by having multiple chunk sizes, but neither is the variance of response time, as shown in Theorem \ref{sec:optVarTThm}.
\begin{theorem}\label{sec:optVarTThm}
The variance of the response time, $\Var{T}$, under Mixed-Random-Chunk is minimized when the number of size-$k_1$ chunks, $a_1$, is either $0$ or $n$.
\end{theorem}
\begin{proof}
By conditioning we have that the variance of the response time is given by
\begin{align*}
\Var{T} =& \E{T-\E{T}}^2 = \E{T^2}-\E{T}^2 \\
=& \frac{a_1}n\E{T^2 \mid E_1} +\frac{a_2}n \E{T^2 \mid E_2}\\
&- \left(\frac{a_1}n\E{T \mid E_1} +\frac{a_2}n \E{T \mid E_2}\right)^2.
\end{align*}

Again, the substitution of $a_2 = n-a_1$ and subsequent algebraic manipulation leads to
\begin{align*}
\Var{T} =& -a_1^2\left(\frac{\E{T \mid E_1} + \E{T \mid E_2}}{n}\right)^2 \\
+& \frac{a_1}{n}\biggl(\E{T^2 \mid E_1} - \E{T^2 \mid E_2}\\
&- 2 \E{T\mid E_1}\E{T\mid E_2} + 2\E{T \mid E_2}^2 \biggr)\\
&+\E{T^2\mid E_2} - \E{T \mid E_2}^2.
\end{align*}

Note that this expression is of the form $\Var{T} = -a_1^2x+a_1y+z$, and thus the variance of the response time is quadratic in $a_1$. Furthermore, it is clear from the expression that $x$ is positive, since mean response times must be positive.  The variance of the response time, then, is a concave quadratic in $a_1$.  Since the number of cores, $a_1$, assigned to size-$k_1$ chunks takes an integer value between 0 and $n$, this means that the variance of the response time is minimized on a boundary, when $a_1=0$ or $a_1=n$.
\end{proof}

\subsubsection*{Why Multiple Chunk Types Do Not Help}
 The issue is that the benefit of having larger chunks available is usually outweighed by a reduction in the stability region of the system.  We have therefore chosen to omit the bulk of our analysis of these policies.

\subsection{The Nelson-Philips Approximation}\label{app-np}
Here we state the approximation of mean response time under JSQ as given in \cite{Nelson1993PerfromEval}.

Let $M/M/c/JSQ$ denote a $c$ core system with total arrival rate $\Lambda$ and service rate $\mu=1/\mathbb{E}[X]$ at each core.  Let 
$$\rho:=\frac{\Lambda}{c\mu}.$$

We then define the following terms:
\begin{align*}
a(\rho)&=1 - \frac{c\rho}{c+4}\\
b(\rho)&=\frac{c\rho}{(c+4)(c-1)}\\
\xi(\rho)&= \frac{\rho(1-c\rho^{c-1}+(c-1)\rho^c)}{(1-\rho)(1-\rho^c)}\\
q(\rho) &= a(\rho) + b(\rho)\xi(\rho)\\
S(\rho)&=\frac{c(1-\rho)}{1-\rho^c}\{\rho^c + q(1-\rho^c)\}\\
\alpha_1 = &0.0455, \alpha_2 = 0.7678, \gamma_1 = 0.0216, \gamma_2= 0.0045\\
r_c&=\gamma_1 \log_2(c) + \gamma_2\\
i_c&=-1/\log_2\{\alpha_1\log_2(c) + \alpha_2\}\\
R(\rho)&= \frac{1}{1-4r_c\rho^{i_c}(1-\rho^{i_c})}\\
A(\rho)&=\left[\sum_{n=0}^{c-1} (c\rho)^n/n!\right] + (c\rho)^c/(c!(1-\rho))\\
P_c(\rho)&=(c\rho)^c/(c!(1-\rho)A(\rho))\\
W_{M/M/c}(\rho)&=\frac{1}{\mu}\frac{P_c(\rho)}{c(1-\rho)}
\end{align*}
The term $W_{M/M/c}(\rho)$ is the mean waiting time in an $M/M/c$ system.  The term $S(\rho)$ is an approximation of the mean length of the shortest queue in the $M/M/c/JSQ$ system.  $R(\rho)$ is an experimentally derived error correction term.  The mean response time in the $M/M/c/JSQ$ system given in \cite{Nelson1993PerfromEval} is:
$$\mathbb{E}[T]^{JSQ} \approx W_{M/M/c}(\rho)S(\rho)R(\rho) + \frac{1}{\mu}.$$

\subsection{Proof of Lemma \ref{lemma:v}}\label{pf:lemma:v}

LEMMA 5.2. \emph{Given any GREEDY policy, $P$, for any $(x_1, x_2) \in \mathcal{S}$, 
$$V^P(x_1+1, x_2) > V^P(x_1, x_2+1).$$
}
\begin{proof}
We begin by defining $V_n^P$ as follows:
$$V_{n+1}^P(x_1, x_2) = A_n^P(x_1, x_2)+I_n^P(x_1, x_2)$$ 
where 
\begin{align*}
A_{n}^P(x_1, x_2) =& c(x_1,x_2)+\Lambda_1\left(V_n^{P}(x_1+1,x_2)-V_n^P(x_1, x_2)\right)\\
&+\Lambda_2 \left(V_n^{P}(x_1, x_2+1)-V_n^P(x_1, x_2)\right)
\end{align*}
and 
\begin{align*}
I_{n}^P=&\gamma_1^P(x_1, x_2)\left(V_n^P((x_1-1)^+, x_2)-V_n^P(x_1, x_2)\right) \\
&+ \gamma_2^P(x_1, x_2)\left(V_n^P(x_1, (x_2-1)^+)-V_n^P(x_1, x_2)\right)
\end{align*}
and $V_0^P(x_1, x_2)=x_1 + x_2 + \frac{x_1}{x_1 + x_2 +1}$ for all $(x_1, x_2) \in \mathcal{S}$.

Recall that $\gamma_i^P(x_1, x_2)$ denotes the departure rate of class $i$ jobs from state $(x_1, x_2)$ under policy $P$, $\Lambda_i$ denotes the arrival rate of class $i$ jobs, and $c(x_1,x_2)=x_1+x_2$.  From \cite{puterman2014markov} we know that 
$$\lim_{n\rightarrow \infty} V_n^P(x_1,x_2) -V_n^P(y_1,y_2) = V^P(x_1, x_2) - V^P(y_1, y_2).$$

Thus, if we can prove that our claim holds for $V_n^P(x_1,x_2)$ for all $n \geq 0$, then it must hold for $V^P(x_1, x_2)$ as well (see, for example, \cite{Koole}).  We will now prove that all three of the following properties of $V_n^P$ hold for all $n\geq 0$ by induction:

\begin{enumerate}
\item $V_n^P(x_1+1, x_2) > V_n^P(x_1, x_2)$ 
\item $V_n^P(x_1, x_2+1) > V_n^P(x_1, x_2)$
\item $V_n^P(x_1+1, x_2) > V_n^P(x_1, x_2+1)$
\end{enumerate}

Note that the first two properties will be necessary for our proof of the third property, and the lemma follows directly from the third property.

We can easily verify that all three properties hold when $n=0$ due to our choice of $V_0^P$.  We now wish to show that if these properties hold for $V_n^P$, they must hold for $V_{n+1}^P$.

To prove property 1, we wish to show that
\begin{align*}
&V^P_{n+1}(x_1+1, x_2) - V^P_{n+1}(x_1, x_2) \\
&= A^P_n(x_1+1, x_2) - A^P_n(x_1, x_2)\\
&\;+I^P_n(x_1+1, x_2) - I^P_n(x_1, x_2) \\
&> 0.
\end{align*}

We can easily see that $A^P_n(x_1+1, x_2) - A^P_n(x_1, x_2) > 0$, since the cost function $c(x_1, x_2)$ is increasing in $x_1$ and $V_n^P(x_1, x_2)$ is increasing in $x_1$ by the property 1 of the inductive hypothesis.  To see that $I^P_n(x_1+1, x_2) - I^P_n(x_1, x_2) > 0$, we can expand the terms as follows:
\begin{align*}
    &I^P_n(x_1+1, x_2) - I^P_n(x_1, x_2) \\
		&=\gamma^P_1(x_1, x_2)\left(V^P_n(x_1, x_2)-V^P_n((x_1-1)^+, x_2)\right) \\
		&\;+\gamma^P_2(x_1, x_2)\left(V^P_n(x_1, x_2)-V^P_n(x_1, (x_2-1)^+)\right) \\
		&\;+\gamma^P_2(x_1+1, x_2)\left(V^P_n(x_1+1, (x_2-1)^+)-V^P_n(x_1, x_2)\right) \\
		&\;+\left(1-\gamma^P_1(x_1+1, x_2)-\gamma^P_2(x_1+1, x_2)\right)\left(V^P_n(x_1+1, x_2)-V^P_n(x_1, x_2)\right).
\end{align*}
Each line here is positive by the inductive hypothesis and the fact that 
$$\left(1-\gamma^P_1(x_1+1, x_2)-\gamma^P_2(x_1+1, x_2)\right) \geq 0$$
since the system was uniformized to one.  Thus property 1 holds.  The proof of property 2 follows a very similar argument.

To show property 3 we wish to show that 
\begin{align*}
&V^P_{n+1}(x_1+1, x_2) - V^P_{n+1}(x_1, x_2+1) \\
&= A^P_n(x_1+1, x_2) - A^P_n(x_1, x_2+1)\\
&\;+I^P_n(x_1+1, x_2) - I^P_n(x_1, x_2+1) \\
&> 0.
\end{align*}
We can easily see that $A^P_n(x_1+1, x_2) - A^P_n(x_1, x_2+1) > 0$ by the inductive hypothesis.  To see that $I^P_n(x_1+1, x_2) - I^P_n(x_1, x_2+1) > 0$, we can expand the terms as follows:
\begin{align*}
&I^P_n(x_1+1, x_2) - I^P_n(x_1, x_2+1) \\
&= \gamma^P_1(x_1, x_2+1)\left(V^P_n(x_1, x_2)-V^P_n((x_1-1)^+, x_2+1)\right) \\
&\;+\gamma^P_2(x_1+1, x_2)\left(V^P_n(x_1+1, (x_2-1)^+)-V^P_n(x_1, x_2)\right) \\
&\;+\left(\gamma^P_1(x_1, x_2+1)+\gamma^P_2(x_1, x_2+1)-\gamma^P_1(x_1+1, x_2)-\gamma^P_2(x_1+1, x_2)\right) \\
&\;\;\;\;\qquad\qquad\times\left(V^P_n(x_1+1, x_2)-V^P_n(x_1, x_2)\right) \\
&\;+\left(1-\gamma^P_1(x_1, x_2+1)-\gamma^P_2(x_1, x_2+1)\right) \\
&\;\;\;\;\qquad\qquad\times\left(V^P_n(x_1+1, x_2)-V^P_n(x_1, x_2+1)\right). \\
\end{align*}
We know, by assumption, that $s_1(k) < s_2(k)$ for $k > 1$, and thus the coefficient $\gamma^P_1(x_1, x_2+1)+\gamma^P_2(x_1, x_2+1)-\gamma^P_1(x_1+1, x_2)-\gamma^P_2(x_1+1, x_2)$ is non-negative.  Therefore, all terms in this sum are positive by the inductive hypothesis, and property 3 holds.

All three properties therefore hold by induction, and $V^P(x_1+1, x_2) > V^P(x_1, x_2+1)$ as desired.
\end{proof}
\section*{Acknowledgements}
The first and third authors were supported by: NSF-CMMI-1538204, NSF-XPS-1629444, by a Google Faculty Research Award, and by a Facebook Faculty Research Award.
Parts of the second author's work were carried out while visiting the Carnegie Mellon University and while affiliated with Leiden University. In addition, the second author's research was in part funded by NWO Gravitation project NETWORKS, grant number 024.002.003.
\bibliographystyle{ACM-Reference-Format}
\bibliography{bibliography}

\end{document}